\documentclass{llncs}
\usepackage[utf8]{inputenc}

\usepackage{amsmath}
\usepackage{amsfonts}
\usepackage{amsthm}
\usepackage{bm}
\usepackage{amssymb}
\usepackage[normalem]{ulem}
\usepackage{enumerate}
\usepackage{tikz,pgfplots}
\usepackage{xcolor}
\usepackage[hidelinks]{hyperref}
\usepackage{cleveref}
\usepackage{mathrsfs}
\usepackage{tikz}
\usetikzlibrary{patterns,arrows,decorations.pathreplacing}
\usetikzlibrary {patterns,patterns.meta} 
\usetikzlibrary{external}
\usepackage{algorithm}
\usepackage{algorithmic}
\usepackage[leftcaption]{sidecap} 
\usepackage{booktabs}
\hypersetup{colorlinks,breaklinks,
             urlcolor=blue,
             linkcolor=blue}
\usepackage{graphicx}

\newtheorem{notation}[theorem]{Notation}
\theoremstyle{remark}

\newcommand{\tw}{\ell}
\renewcommand{\P}{\mathcal{P}}

\newcommand{\ev}{\mathrm{ev}}
\newcommand{\evav}{\mathrm{ev}_{\bm{\a},\bm{v}}}

\newcommand{\F}{\mathbb{F} }

\newcommand{\Z}{\mathbb{Z} }
\newcommand{\RS}{\mathbf{RS}}
\newcommand{\GRS}{\mathbf{GRS}}

\renewcommand{\c}{\mathcal{C}}
\newcommand{\CC}{\mathcal{C}}
\newcommand{\R}{\mathcal{R}}
\newcommand{\EE}{\mathcal{E}}

\newcommand{\Span}[1]{\left\langle #1 \right\rangle}

\newcommand{\citep}{\cite}

\newcommand{\Fq}{\F_q}
\newcommand{\Fqx}{\Fq[x]}
\newcommand{\Fqn}{\F_q^n}
\newcommand{\Fqxk}{\Fq[x]_{<k}}

\newcommand{\<}{\left<}
\renewcommand{\>}{\right>}

\newcommand{\C}{\mathcal{C}}
\renewcommand{\r}{\rho}

\newcommand{\G}{\mathbf{G}}

\newcommand{\D}{\mathcal{D}}
\renewcommand{\a}{\alpha}
\newcommand{\aav}{\bm{\a}}

\newcommand{\A}{\mathcal{A}}

\newcommand{\Gpub}{\mathbf{G}_{pub}}
\newcommand{\MM}{\mathcal{M}}
\newcommand{\Mon}[1]{\Fqx_{<#1}}
\newcommand{\Monh}[2]{\text{Mon}_{<#1,\hat{#2}}}

\newcommand{\eqdef}{\stackrel{\textrm{def}}{=}}

\newcommand{\B}{\mathcal{B}}

\renewcommand{\aa}{\boldsymbol{a}}
\newcommand{\bb}{\boldsymbol{b}}

\renewcommand{\le}{\leqslant}
\renewcommand{\ge}{\geqslant}

\newcommand{\vv}{\mathcal{\textbf{v}}}
\newcommand{\E}{\EE'}
\newcommand{\schur}{\star}

\newcommand{\vect}[1]{\boldsymbol{#1}}

\newcommand{\bv}{\vect{b}}
\newcommand{\cv}{\vect{c}}

\renewcommand{\vv}{\vect{v}}
\newcommand{\xv}{\vect{x}}
\newcommand{\yv}{\vect{y}}

\newcommand{\map}[4]{
\left\{
    \begin{array}{ccc}
        #1 & \longrightarrow & #2  \\
        #3 & \longmapsto     & #4
    \end{array}
\right.
}
\renewcommand{\geq}{\geqslant}
\renewcommand{\leq}{\leqslant}

\title
{On the structure of the Schur squares of Twisted Generalized Reed-Solomon codes and application to cryptanalysis} 
	\author{
	    Alain Couvreur\inst{1,2} \and
	    Rakhi Pratihar\inst{1,2} \and
	    Nihan Tan{\i}sal{\i}\inst{1,2}\and
	    Ilaria Zappatore\inst{3}
	    }
\institute{
    Inria \and
    Laboratoire LIX, CNRS UMR 7161, École Polytechnique, Institut Polytechnique de Paris, 1 rue Honoré d'Estienne d'Orves, 91120 Palaiseau Cedex \and
   XLIM, CNRS UMR 7252, Université de Limoges, 123, avenue Albert Thomas - 87060 Limoges Cedex
   \email{\{alain.couvreur,rakhi.pratihar,nihan.tanisali\}@inria.fr}\\
   \email{ilaria.zappatore@unilim.fr}
}

\begin{document}
\maketitle
\begingroup
 \renewcommand{\thefootnote}{}%
\footnotetext{This work is partially funded by the French Agence Nationale de la Recherche project ANR-21-CE39-0009-BARRACUDA, by Plan France 2030 ANR-22-PETQ-0008 and by Horizon-Europe MSCA-DN project ENCODE.}
\endgroup
\addtocounter{footnote}{-1}
\begin{abstract}
Twisted generalized Reed-Solomon (TGRS) codes constitute an interesting family of evaluation codes, containing a large class of maximum distance separable codes non-equivalent to generalized Reed-Solomon (GRS) ones.
Moreover, the Schur squares of TGRS codes may be much larger than those of GRS codes with same dimension.
Exploiting these structural differences, in 2018, Beelen, Bossert, Puchinger and Rosenkilde proposed a subfamily of Maximum Distance Separable (MDS) Twisted Reed--Solomon (TRS) codes over $\Fq$ with $\ell$ twists $q \approx n^{2^{\ell}}$ for McEliece encryption, claiming their resistance to both Sidelnikov Shestakov attack and
Schur products--based attacks. In short, they claimed these codes to resist to classical key recovery attacks on McEliece encryption scheme instantiated with Reed-Solomon (RS) or GRS codes. In 2020, Lavauzelle and Renner presented an original attack on this system based on the computation of the subfield subcode of the public TRS code.

In this paper, we show that the original claim on the resistance of TRS and TGRS codes to Schur products based--attacks is wrong.
We identify a broad class of codes including TRS and TGRS ones that is distinguishable from random by computing the Schur square
of some shortening of the code. Then, we focus on the case of single twist (\emph{i.e.}, $\ell = 1$), which is the most efficient one in terms of decryption complexity, to derive an attack. The technique is similar to the distinguisher-based attacks of RS code-based systems given by Couvreur, Gaborit, Gauthier-Umaña, Otmani, Tillich in 2014. 
\end{abstract} 


\section{Introduction}

\emph{McEliece's encryption scheme} dates back to the early ages of public key cryptography. For a long time it has been considered unusable because of the
significant size of the public key. However, with the recent and growing interest
of post-quantum cryptographic primitives, Classic McEliece \cite{ABCCGLMMMNPPPSSSTW20}
could be standardized in the near future. Besides the seminal proposal by McEliece
himself based on classical Goppa codes, there have been many attempts to replace Goppa codes by other families of codes with more efficient decoding algorithms in order to reduce the key size. In 1986, Niederreiter \cite{Nie86} suggested generalized Reed-Solomon codes (GRS) to replace Goppa codes, but it was shown to be insecure by Sidelnikov and Shestakov in \cite{SS92}. Thereafter, several instantiations using codes
``close to GRS" codes appeared. Berger and Loidreau replaced the GRS code by a subcode of low codimension \cite{BL05}; Wieschebrink \cite{Wie06} included some random columns in a generator matrix of a GRS code; this approach was further enhanced in \cite{Wan16,Wan17} by additionally ``mixing'' the random columns with the original ones via specific linear transformations; finally, in \cite{BBC+} Baldi, Bianchi, Chiaraluce, Rosenthal, and Schipani proposed to mask the structure of a GRS code by
right multiplying it by a ``partially weight-preserving'' matrix. All these proposals have been partially or fully broken using attacks derived from a square code distinguisher \cite{CGGOT,CLT19,COTG15,Wie10}. 

\emph{Twisted Reed-Solomon codes} (TRS) are evaluation codes in the Hamming metric. These codes were introduced in \cite{Beelen}, adapting to the Hamming setting Sheekey's construction of \emph{twisted Gabidulin codes} \cite{Sheekey} in rank metric.
Unlike Reed-Solomon codes, TRS codes are not always maximum distance separable (MDS) codes, but this family of codes contains MDS codes that are not equivalent to generalized Reed-Solomon (GRS) codes. 
Recently, in \cite{BBPR}, TRS codes have been proposed as an alternative to Goppa codes for McEliece cryptosystem, and example parameters are given that provide shorter keys compared to the original McEliece cryptosystem for the same security level. The authors also singled out a subfamily of twisted Reed-Solomon codes, which they ``provably'' claimed to be resistant against several known structural attacks on the McEliece
cryptosystem based on RS-like codes: Sidenlikov–Shestakov \cite{SS92}, Wieschebrink \cite{Wie,Wie10}, Schur square-distinguishing \cite{CGGOT}.

More recently, Lavauzelle and Renner \cite{LR} gave an efficient key-recovery attack on the TRS variant proposed in \cite{BBPR} based on identifying some specific structure of the \emph{subfield subcode}. Lavauzelle and Renner assumed the security claims of \cite{BBPR} on the resistance of TRS codes to Schur squares to be true and identified another weakness that was very specific to the chosen public keys (\textit{i.e.} the underlying TRS code is MDS) coming from some tower of extensions of finite fields. It is worth mentioning that their attack restricts to a
limited subfamily of TRS codes.

In this paper, \emph{our contributions} are threefold.
First, we show the claims of \cite{BBPR} on the resistance of TRS codes to Schur square distinguishing are wrong. This holds even for their generalized version : TGRS codes. 
Hence we are able to prove that such codes are distinguishable from random as soon as their
number $\ell$ of twists is in $O(1)$. The latter assumption is reasonable since the
decoding of such codes is exponential in $\ell$. In short, we prove that any TGRS codes
that could be proposed for the McEliece scheme are actually distinguishable from random.

Second, for the case of TGRS codes with a single twist ($\ell = 1$) we show how to derive
a polynomial time attack from the Schur square distinguisher. This attack is in the very same flavour as
the attack of \cite{CGGOT,COTG15} on BBCRS scheme \cite{BBC+} and runs in $O(q^3 n^4)$
operations in $\Fq$. Note that the family of
codes for which such an attack applies is much larger than the family broken by \cite{LR},
which was restricted to a restricted family of TRS codes.

Third, the attacks in \cite{CGGOT,COTG15} involved a heuristic argument that was claimed to hold ``with a high probability". In the present article, we provide a detailed analysis of the success probability of the algorithm, providing a proven attack of the scheme.
{Note that we are able to estimate this success probability for a range of parameters that is \textbf{strictly included} in the range of parameters we can actually attack. In short : this attack works on a broad range of parameters and we can prove its success without involving any heuristic arguments in some subrange of parameters.
}

\subsection*{Outline of the article}
Section~\ref{sec:prelim} provides the basic notation and the necessary prerequisites on
GRS codes, TGRS codes, Schur products, and McEliece encryption scheme.
In Section~\ref{sec:qGRS}, we introduce a new class of codes called \emph{quasi--GRS codes} which strictly includes
TGRS codes, which turns out to be the class of codes we succeed distinguishing from random
by computing the Schur square of some of their shortening.
In Section~\ref{sec:distinguisher}, we show how to distinguish quasi--GRS codes from random.
Section~\ref{sec:5} is dedicated to the presentation of an attack on the McEliece
instantiated with TGRS codes with a single twist and a \emph{SageMath} implementation of the  attack together with timing results are given in Section~\ref{sec:implem}. Finally, Section~\ref{pol}
is dedicated to a probability analysis yielding the proof of a crucial theorem 
for the attack.



\section{Preliminaries}\label{sec:prelim}

Let $q$ be a prime power, $\Fq$ be a finite field of order $q$, and $\Fqn$ the $\Fq$ vector space of dimension $n$. In this article, vectors are represented by lowercase bold letters: $\boldsymbol{a}, \boldsymbol{b}, \boldsymbol{c}$ and matrices by uppercase bold letters $\G, \mathbf{H}$.
Given a positive integer $n$, we denote $[n]$ to be the set $[n]\eqdef\{1,\ldots,n\}$.
We denote by $\Fqxk$ the space of polynomials over $\Fq$ of degree strictly less than $k$.
Finally, given elements $a_1,\dots, a_s$ of a given vector space $V$, we denote by $\Span{a_1,\dots, a_s}$
the vector space spanned by these elements. Note that all the considered vector spaces in this paper
are over $\Fq$, hence the field is not specified when mentioning dimension or vector span.




Sometimes, since we work with \textit{Twisted Generalized Reed Solomon codes}, it would be useful to remove one monomial $x^h$ (for a certain $0\leq h \leq k-1$) from the previously defined subspace $\Fqxk$. In that case, we denote:
\begin{equation}
\Monh{k}{h}\eqdef \Span{x^0, \ldots, \widehat{x^h}, \ldots, x^{k-1}},
\end{equation}
where the hat notation means that the monomial $x^h$ is removed.
In a more general situation, when we remove more than one monomial, for instance $x^{h_1}, \ldots, x^{h_{\ell}}$ for $\ell \geqslant 1$, we denote,
\begin{equation}\label{eq:monCod1Space}
\Monh{k}{\boldsymbol{h}}\eqdef \Span{x^i : i \in [n]\setminus \{h_1, \ldots, h_{\ell}\}}
\quad \text{where}\quad \boldsymbol{h}\eqdef(h_1, \ldots, h_{\ell}).
\end{equation}
Finally, the present article crucially uses the notion of 
\emph{shortening} of codes:

\begin{definition}
    Given a code $\C \subseteq \Fq^n$ and a subset $I = \{i_1, \ldots, i_{|I|}\} \subseteq [n]$, the \emph{shortening}  of $\C$ at $I$, denoted as $\C_I$, is defined as
    \[
       \C_I \eqdef \{(x_i)_{i \in [n] \setminus I} ~:~ \xv =(x_1, \ldots, x_n) \in \C  ~\text{such\ that}~ \forall i \in I,~ x_i = 0\}.
    \]
\end{definition}

\begin{remark}
In short, the shortened  code is the subcode of vectors whose entries indexed by $I$ are zero and whose prescribed zero entries are removed. Then at some places in the paper -- we will mention it when needed -- we will not remove the prescribed zero entries from the shortening and hence take as a definition:
\[
    \{\xv \in \C ~:~ \forall i \in I,\ x_i = 0\}.
\]
\end{remark}

\subsection{Schur products of linear spaces and evaluation codes}

One of the most important tools in the cryptanalysis of the McEliece encryption scheme and its variants is the Schur product of codes. For this reason, we introduce the following definitions.

\begin{definition}[Componentwise product] Given $\aa = (a_1, \ldots, a_n)$ and $\bb=(b_1, \ldots, b_n)$ in $\Fqn$, we denote by $\aa\schur\bb$ the \emph{componentwise} or \emph{Schur} product as
$$
a\schur b \eqdef (a_1b_1, \ldots, a_n b_n).
$$
\end{definition}

\begin{definition}[Schur product of linear codes and square code]\label{def:schurCode} The
\emph{Schur product} of two linear codes $\A,\ \B \subset \Fq^n$ is defined as,
$$
\A \schur \B \eqdef \langle \{ \aa \schur \bb : \aa \in \A, \bb \in \B\}\rangle
$$
When $\A=\B$ then $\A \schur \A$ is called \emph{square} of $\A$ and is denoted by $\A^2$.
\end{definition}
We can easily derive an upper bound on the dimension of the Schur product of two codes and of the square product code : if we consider a basis of $\A$ and $\B$, the product space is generated by the componentwise products of their elements.

\begin{proposition}\label{prop:randomsquare}
    If $\A$ and $\B$ two linear codes of length $n$. Then,
    \begin{enumerate}
        \item $\dim(\A \schur \B)\leq \dim(\A)\dim(\B)$,
        \item $\dim(\A^2) \leq \binom{\dim(\A)+1}{2}$
    \end{enumerate}
\end{proposition}
In particular, for a random code of dimension $k$ and length $n$, it can be shown that the dimension of the square is  $\min\left\{\frac{k(k+1)}{2},n\right\}$ with high probability. Such a statement is proved in \cite[Thm.~2.3]{CCMZ15} for binary codes
and in \cite[Thm.~2]{FGVOPT13} for $q$-ary ones.
\begin{proposition}\label{prop:randomsquaredim}
    Let $k,n\geq 0$ such that ${k+1\choose 2} < n$ and $\A$ be a random $[n,k]$ code. We have
    $$
    \text{Prob}\left[ \dim \A^2 < \binom{k+1}{2}\right]=o(1).
    $$
\end{proposition}


\subsection{Generalized Reed-Solomon codes and the square code construction }

Given a vector $\bm{\a} = (\a_1, \ldots, \a_n) \in \Fq^n$ for distinct $\a_i$'s and $\bm{v} = (v_1, \ldots, v_n)$ a nonzero vector in $\Fqn$, we consider the following \textit{evaluation map} associated to $\bm{\a}$ and $\bm{v}$,
\begin{equation}\label{eval}
\mathrm{ev}_{\bm{\a},\bm{v}} \colon \map{\Fq[x]}{\Fq^n}{f(x)}{(v_1 f(\a_1), \ldots, v_n f(\a_n)).}
\end{equation}
\begin{definition}[Generalized Reed-Solomon Code]
    Let $k, n$ be positive integers, $k < n\leq q$. The $[n,k]_q$ \emph{generalized Reed-Solomon (GRS)} code associated with $\bm{\a}, \bm{v}$ is defined as
     \begin{align*}
        \GRS_k(\bm \a, \bm{v})\eqdef \{( \mathrm{ev}_{\bm{\a}, \bm{v}}(f) ~\colon~ f \in \Fqxk)\}.
     \end{align*}
    If $\bm{v} = (1, \dots, 1)$, then the code is a \emph{Reed--Solomon code} denoted as $\RS_k(\mathbf{\bm{\alpha}})$.
\end{definition}

GRS codes lie at the core of algebraic coding theory. Many other algebraic constructions of codes, such as Alternant codes, BCH codes, Goppa codes, Srivastava codes, and so on, are derived from GRS ones. GRS codes are famous because of their optimal parameters: they are known to be \emph{Maximum distance Separable} (MDS), \emph{i.e.} they have the best possible minimum distance with respect to their length and dimension. In addition, such codes benefit from efficient decoding algorithms up to half their minimum distance \cite[Chapter~6]{Roth} and even beyond using list decoding \cite{GS98,Sudan97}.

\medskip
The Schur product (see Definition~\ref{def:schurCode}) plays a central role in the cryptanalysis of the McEliece cryptosystem and its variants. It indeed permits to distinguish algebraically \textit{structured} codes such as GRS ones from random codes.
\begin{proposition}\label{prop:sq_GRS}
    $\GRS_k(\bm \a, \bm{v}) \star \GRS_{\ell}(\aav, \vv) = \GRS_{k+\ell-1}(\bm \a, \bm{v}\star \vv)$.
\end{proposition}

\begin{proof}
See \cite{CGGOT}.
\end{proof}

Given a GRS code of dimension $k$, Proposition~\ref{prop:sq_GRS} entails that the dimension of its square is $2k-1$ while Proposition~\ref{prop:randomsquaredim} asserts that it would be $\min\{n, \frac{k(k+1)}{2}\}$ for a random code.
Thus, GRS codes of dimension $k < \frac n 2$ can be easily distinguished from random ones by computing their Schur square.
For GRS codes of dimension $\geq \frac n 2$, they can be distinguished from random by computing the square of their duals which is itself
a GRS code.

Sometimes, the sole use of the Schur product is insufficient to provide a distinguisher, it is then useful to consider the square of a shortening of the the code. Shortening then squaring turns out to be very efficient to distinguish codes ``close" to GRS codes, as shown in \cite{CGGOT,CL22,CLT19,CMP17,COT14a,COT17,COTG15}. For this reason, we introduce the following Lemma that will be useful later.
\begin{lemma}
    Shortening of an $[n,k]$ GRS code at $a \le k$ positions gives an $[n- a, k - a]$ GRS code.
\end{lemma}

\subsection{Twisted generalized Reed--Solomon codes}
\emph{Twisted Reed--Solomon codes} (TRS) and \emph{Twisted Generalised Reed--Solomon} (TGRS) codes are slight variants of GRS codes. 
First inspired by a rank--metric counterpart \cite{Sheekey},
they have been first introduced in \cite{BPR17,Beelen} as a possible alternative to GRS codes containing some MDS codes.
Next, they have been proposed for cryptographic applications in \cite{BBPR}
with the argument that their structure was better hidden with respect to 
classical attacks such as Sidelnikov Shestakov attack or Schur square attack.
They are defined as follows.

\begin{definition}[Twisted generalized Reed-Solomon Codes]\label{def:TGRS}
    For positive integers $n,k, \tw $ with $\ell \le k \le n \le q$, suppose that $\mathbf{h}= (h_1, \ldots, h_\tw )  \in \{ 0,\ldots, k-1 \}^{\tw}$, $ \mathbf{t}= (t_1, \ldots, t_\tw )  \in \{ 1,\ldots, n-k \}^{\tw} $ and $ \boldsymbol{\eta}= (\eta_1, \ldots, \eta_\tw )  \in \F_q ^{\tw} $. Then 
     \begin{align}\label{twistedp}
         \mathcal{P}_{\mathbf{t},\mathbf{h},\bm{\eta}} ^{n,k} \eqdef \left\{ f=   \sum_{i=0 }^ {k-1}     f_i x^i +    \sum_{j=1}^ { \tw} \eta_j f_{h_j}  x^{k-1+t_j} ~:~ f_i \in \F_q  \right\}
     \end{align}
    is a $k$-dimensional $\Fq$-subspace of $\Fq[x]$. Furthermore, let $\bm{\a}= (\a_1,\ldots,\a_n) \in \Fq^n$ where $\a_i$, for $i =1, \ldots, n$ are distinct and $\bm{v} = (v_1, \ldots, v_n) \in (\Fq^{\times})^{n}$.
  Then the corresponding linear code is defined as
     \begin{align*}
         \mathcal{C}_{\bm{\a},\bm{v},\mathbf{t},\mathbf{h},\bm{\eta}}^{n,k} \eqdef \mathrm{ev}_{\bm{\a},\bm{v}} (\mathcal{P}_{\mathbf{t},\mathbf{h},\bm{\eta}} ^{n,k}) \subset \F_q^n.
     \end{align*} 
    is called $(\bm{\a},\bm{v},\mathbf{t},\mathbf{h},\bm{\eta})$-\emph{twisted generalized Reed-Solomon} (TGRS) code. For bre\-vity, we will simply use $\P$ and $\C $ to denote the space of twisted polynomials $\mathcal{P}_{\mathbf{t},\mathbf{h},\bm{\eta}}^{n,k}$ and the TGRS code  $\mathcal{C}_{\bm{\a}, \bm{v}, \mathbf{t},\mathbf{h},\bm{\eta}} ^{n,k}$. 
\end{definition}

The integer $\tw$ is referred to as the \emph{number of twists}, and we call the vectors $\bm{t}, \bm{h}$ and $\bm{\eta}$ the \textit{twist} vector, \textit{hook} vector, and \textit{coefficient} vector, respectively.

In \cite{BBPR}, the authors also proposed a \textit{brute force} technique to decode Twisted Reed--Solomon codes, which can be easily generalized to TGRS. Given a received word $\boldsymbol{y}=\boldsymbol{c}+\boldsymbol{e}\in \Fqn$, where $\boldsymbol{c}\in \C$, the main idea consists in guessing $\ell$ elements $g_1, \ldots, g_{\ell} \in \Fq$ and then decoding $\boldsymbol{y}-\ev_{\bm{\a}, \boldsymbol{v}}\sum_{j=1}^{\ell} g_j\eta_ix^{k-1+t_j}$ as a received word of a $\GRS_k(\bm{\alpha})$. This strategy succeeds if $g_j=f_{h_j}$ for all $1\leq j\leq \ell$. The complexity of such a decoder is $q^{\ell}$ times the complexity of a GRS decoder, and the decoding radius is the same as the GRS one. 

In \cite{Beelen}, the authors proposed a decoding strategy for TRS codes based on a \textit{key equation} which can eventually achieve a better complexity than the brute force. In particular, they introduced a \textit{partial unique decoder},  which means that there are some error patterns that cannot be corrected. They also provide heuristic estimations of the \textit{failure probability} of the decoder and of the decoding radius, which can be smaller than half of the minimum distance of the code. However, even if this algorithm achieves a better complexity than brute force, it is still exponential in the number of twists.

\subsection{Reasoning on the level of polynomial spaces}\label{ss:poly}
All the codes involved are derived from polynomial evaluations. Note that, given the evaluation map $\evav$, then by interpolation, any element $\cv \in \Fq^n$
can be seen as $\cv = \evav(f)$ for a unique $f \in \Fq[x]_{<n}$. Consequently, we will often reason either at the level of codewords, or at the level of polynomials. The correspondence is as follows: for any code $\C \subseteq \Fq^n$, there exists a unique subspace $\mathcal{P}_\C \subseteq \Fq[x]_{<n}$
such that 
\begin{equation}\label{eq:evC}
\C = \evav(\mathcal{P}_\C).
\end{equation}
This correspondence will allow us to move between codewords and polynomials as needed.
Similarly to codes, we define products of polynomial spaces as follows. Given, $\mathcal{P}, \mathcal{R} \subseteq \Fq[x]$ we define
\[
    \mathcal P \mathcal R \eqdef \Span{fg ~:~ f \in \mathcal P, g\in \mathcal R}.
\]
When $\mathcal P = \mathcal R$ we call this the \emph{square} of $\mathcal P$ and denote it $\mathcal{P}^2$. Next, given a polynomial $f \in \Fq[X]$ we use $f \mathcal P$ for $\Span{f}\mathcal P$.

As in the code setting (see Proposition~\ref{prop:randomsquare}), we have natural upper bounds.

\begin{proposition}\label{prop:dim_prod_space_poly}
    Let $\mathcal P$ and $\mathcal R$ be two subspaces of $\Fq[x]$,
    then \[
    \dim \mathcal P \mathcal R \leq \dim \mathcal P \cdot \dim \mathcal R \quad \text{and}\quad
    \dim \mathcal{P}^2 \leq {\dim \mathcal P + 1 \choose 2}.
    \]
\end{proposition}

When considering space of polynomials of bounded degree, as in the case of GRS codes, we also have an 
explicit description of the product.

\begin{proposition}\label{prop:prod_space_poly_bounded}
\(
    (\Fq[x]_{<k}) (\Fq[x]_{< \ell}) = \Fq[x]_{<k + \ell - 1}.
\)
\end{proposition}

\subsubsection{Relation with Schur squares.}
The previously introduced map
\[
    \evav : \map{\Fq[x]}{\Fq^n}{f}{(f(\alpha_1), \dots, f(\alpha_n))}
\]
is a surjective ring morphism when $\Fq^n$ is equipped with the Schur product $\star$. Consequently, from the setup in \eqref{eq:evC}, we obtain a surjective map
\(\mathcal{P}_{\C}^2 \rightarrow \C^2\) which yields:

\begin{proposition}\label{prop:bounding_dim_prod}
Let $\C = \evav(\mathcal{P}_\C)$ for $\mathcal{P}_\C \subseteq \Fq[x]_{<n}$. Then
\[
\dim \C^2 \leq \dim \mathcal{P}_{\C}^2.
\]
\end{proposition}

\subsubsection{Relations with shortenings.}\label{sss:shorten_poly}
Given a code $\C = \evav(\mathcal{P})$ for some $\mathcal P \subseteq \Fq[x]$
and $I \subseteq [n]$, the shortening $\C_I$ of $\C$ at $I$ corresponds to
the evaluation of $\mathcal{P} \cap (p_I)$ where $(p_I)$ is the ideal spanned
by 
\[
    p_I (x) \eqdef \prod_{i \in I}(x - \alpha_i).
\]
In the sequel, we denote
\[
    \mathcal{P}_{p_I} \eqdef \mathcal{P} \cap (p_I).
\]

\subsection{The McEliece cryptosystem and its variants}\label{sec:McEliece}

McEliece code-based original cryptosystem was introduced in 1978 by McEliece \cite{McEliece} and it used binary Goppa codes. However, it corresponds to
very general framework that can be instantiated with many possible codes.

Consider a family of codes $\mathcal F$ parameterized by a set $\mathcal S$. Hence, we are given a map
\[
    \C : \map{\mathcal S}{\mathcal F}{s}{\C (s).}
\]
The above map is the trapdoor: $\C(s)$ should be easy to compute from the knowledge of $s$ but given $\C \in \mathcal F$, recovering $s \in \mathcal S$
such that $\C = \C (s)$ should be hard.
Moreover, suppose that for any $s \in S$ we are given a decoder $D(s)$
that corrects up to $t$ errors for the code $\C (s)$. Here again, decoding
should not be possible without the knowledge of $s$. Then, McEliece
encryption scheme can be described as follows:


\paragraph{Key generation.}
Draw $s \in \mathcal S$ at random. The secret key is $s$.\\ The public
key is a pair $(\G_{pub}, t)$, where $\G_{pub}$ is a generator matrix of $\C (s)$ and
$t$ is the number of errors that the algorithm $D(s)$ can decode.
\paragraph{Encryption.} To encrypt a plaintext $\mathbf{m} \in \Fq^k$,
choose a random $\mathbf{e}\in \Fqn$ with Hamming weight $wt_H(\mathbf{e}) = t$. The ciphertext is $\mathbf{c} \eqdef \mathbf{m} \mathbf{G}_{pub} + \mathbf{e}.$

\paragraph{Decryption.} Apply
the decoder $D(s)$ to the ciphertext $\cv$ to recover $\mathbf{m}$.

\subsubsection*{Examples of instantiations.}
McEliece's original proposal \cite{McEliece} is instantiated with classical Goppa codes, with the secret being the pair $(L, g)$ where $L$ is the so--called \emph{evaluation sequence} and $g$ is the \emph{Goppa Polynomial}. Later on,
Niederreiter \cite{Nie86} proposed to instantiate it with GRS codes, where the secret is $(\xv, \yv)$. This proposal was subsequently broken in \cite{SS92}.
Since then, several alternative instantiations have been suggested (the following list is far from being exhaustive): Reed-Muller codes \cite{S94}, Algebraic Geometry codes and their subfield subcodes \cite{JM96}, subcodes of GRS codes
\cite{BL05}, MDPC codes \cite{MTSB13}, among others. Several other proposals 
rely on slightly modified GRS codes; see, for instance, \cite{BBC+,KRW21,Wan16,Wie06}. However, many of the instantiations
based on codes ``close" to GRS codes have been vulnerable to attacks, with many such attacks involving the Schur product 
\cite{CGGOT,CL22,CLT19,CMP17,COT14a,COT17,Wie10}

Recently, Beelen, Bossert, Puchinger, and Rosenkilde proposed TRS codes in \cite{BBPR} claiming that the corresponding cryptosystem is resistant to the well-known attacks \cite{CGGOT,SS92,Wie}.

\subsection{Cryptanalysis of the McEliece system based on GRS codes and their variants}\label{subsec:cryptanalysisTGRS}
The TRS variant of the McEliece cryptosystem was already broken in a prior work by Lavauzelle and Renner \cite{LR}, which examined a setup different from ours. 
More precisely,
\begin{enumerate}
    \item they considered a prime power $q_0$, the integers $k<n\leq q_0-1$ with $2\sqrt{n}+6<k\leq n/2-2$ and a twist $\ell$ such that,
    \[
    \frac{n+1}{k-\sqrt{n}}<\ell+2<\min\{k+3, 2n/k, \sqrt{n}-2\}.
    \]
    Further, they set $q_i \eqdef q_{i-1}^2 = q_0^{2^i}$
for $i = 1, \ldots , \ell$, such that $\F_{q_0} \subset \F_{q_1} \subset \cdots \subset \F_{q_\ell} = \Fq$ is a chain of subfields. And finally they take $t_i = (i + 1)(r - 2) - k + 2$ and $h_i = r - 1 + i$ for $i = 1, \ldots , \ell$,
where $r \eqdef  \lceil \frac{n+1}{\ell +2}
\rceil + 2$.

Note that this choice of parameters guarantees the underlying TRS to be MDS \cite{BBPR}.
\item They assumed that the integers $q_0, n, k, \ell$ and the hook and twist vector $\mathbf{h},\mathbf{t}$ are \textit{public parameters.}
    
\end{enumerate}
Lavauzelle and Renner's approach \textbf{assumes the validity of claims in \cite{BBPR} regarding the indistinguishability of such codes} with respect to the Schur square even after shortening.
So they use a strategy based on the \textbf{recovery of the subfield subcode} to attack the system. The novelty of their approach lies in the fact that, in this case, they use the subfield subcode structure of the TRS to attack the system, whereas usually the subfield subcode operation is used to hide the structure of the code used in encryption to improve the security.
In particular, thanks to the previous choice of parameters, they describe the structure of the subfield subcode (in $\mathbb{F}_{q_0}$) as a subspace of low codimension of a classical RS code, and they exploit this code to recover the hidden TRS code.

In contrast, our work considers a \textbf{general TGRS family} without specific parameter assumptions, and we demonstrate that these codes, like GRS codes, can indeed be distinguished from random codes. Then, we extend the distinguisher-based attack of \cite{CGGOT} to such a code family.




\section{$\ell$--quasi--GRS codes}\label{sec:qGRS}
In this section, we introduce a broader class of codes called \emph{quasi--GRS} codes, which contains TGRS codes.
Notably, the distinguisher we describe further, along with most of the cryptanalysis techniques presented,
applies to quasi--GRS codes. The interest of this class is that it is closed under duality
and ``most of the times", closed under shortening.

\begin{definition}\label{def:qGRS}
Let $\aav \in \Fq^n$ be a sequence of distinct elements and $\vv \in (\Fq^\times)^n$. An $\ell$--quasi--GRS ($\ell$--qGRS) code is defined as
a code $\C$ such that
\[
    \C = \C_0 \oplus \C_1,
\]
where $\C_0$ is a subcode of codimension $\ell$ of $\GRS_k(\aav,\vv)$ and $\C_1$ has dimension $\ell$ and satisfies
$\C_1 \cap \GRS_k(\aav, \vv) = 0$.
\end{definition}

\begin{proposition}\label{prop:TGRS=qGRS}
A TGRS code with $\ell$ twists is an $\ell$--qGRS code.
\end{proposition}

\begin{proof}
Using notation from Definition \ref{def:TGRS}, define 
\begin{align*}
\C_0 &\eqdef \ev (\Span{x^i ~:~
i \in [n] \setminus \{h_1, \dots, h_\ell\}}) = \ev(\Monh{k}{\bm{h}})  \\
\text{and}\quad \C_1 &\eqdef 
\ev (\Span{x^{h_1}+\eta_1 x^{k-1+t_1},\dots, x^{h_\ell} + \eta_\ell x^{k-1+t_\ell}}).
\end{align*}
This yields the result.
\end{proof}

\begin{proposition}
The dual of an $\ell$--qGRS code is an $\ell$--qGRS code.
\end{proposition}

\begin{proof}
Let $\C = \C_0 \oplus \C_1$ be an $\ell$--qGRS code with $\C_0 \subseteq \GRS_k(\aav,\vv)$ and $\C_1 \cap \GRS_k(\aav, \vv) = 0$. Then, $\C$ has codimension $\ell$ in $\GRS_k (\aav, \vv) \oplus \C_1$.
Therefore, $\GRS_k(\aav, \vv)^\perp \cap \C_1^\perp$ has codimension $\ell$
in $\C^\perp$. Denote by $\C_0' \eqdef \GRS_k(\aav, \vv)^\perp \cap \C_1^\perp$ and let $\C_1'$ be a complement subspace of $\C_0'$ in $\C^\perp$. Therefore, it is clear that $\C_0' \oplus \C_1' = \C^{\perp}$ and $\dim(\C_1')=\ell$. Note that $\C_1' \cap \GRS_k(\aav,\vv)^\perp = \{0\}$, as otherwise, $ \GRS_k(\aav,\vv)^\perp \cap \C_1' \subseteq \GRS_k(\aav,\vv)^\perp \cap \C^{\perp} = \C_0'$ will lead to a contradiction. Now, since $\GRS_k(\aav,\vv)^\perp$ is again a GRS code and $\C'_0$ has codimension $\ell$
in $\GRS_k(\aav,\vv)^\perp$, we can conclude that $\C^{\perp}$ is an $\ell$--qGRS code.
\end{proof}

\begin{proposition}\label{prop:shorten_qGRS}
Let $\C$ be an $\ell$--qGRS code of dimension $k$ and $I \subseteq [n]$ with $|I| \leq k - \ell$.
If $\dim \C_I = k - |I|$, then $\C_I$ is an $\ell'$--qGRS code for some
$\ell' \leq \ell$.
\end{proposition}

The proof of Proposition~\ref{prop:shorten_qGRS} rests on the following lemma.

\begin{lemma}\label{lem:short_superGRS}
Let $\D = \GRS_k(\aav, \vv) \oplus \D_1$ for some code $\D_1$ of dimension $\ell$.
Let $I \subseteq [n]$ such that $|I| < k$. Denote $s\eqdef |I|$. Then, $\dim \D_I = \dim \D - s$ and is of the form $\D_I = \GRS_{k-s}(\aav', \vv') \oplus \D_1'$
for some $\aav',\vv'\in \Fq^{n-s}$ and some code $\D_1' \in \Fq^{n-s}$ of dimension $\ell$.
\end{lemma}

\begin{proof}
Without loss of generality, by possibly permuting the entries of codewords,
one can assume that $I = \{1, \dots, s\}$.
Since $\GRS_k(\aav, \vv)$ is MDS, it has a systematic generator matrix:
\(
    \begin{pmatrix}
        \mathbf{I}_k & \big| & \mathbf A
    \end{pmatrix}
\)
for some $\mathbf A \in \Fq^{k\times(n-k)}$.
Then, by elimination, $\D$ has a generator matrix of the form
\[
    \begin{pmatrix}
        \mathbf{I}_k & \mathbf A \\
        (0) & \mathbf B
    \end{pmatrix}
\]
for some matrix $\mathbf B \in \Fq^{\ell \times (n-k)}$ in row echelon form.
Let $\G$ be the matrix obtained by removing the first $s$ rows and columns of the above matrix. The matrix $\G$ is nothing but a generator matrix of $\D_I$.
Since $\mathbf B$ is in row echelon form, we have the result on the dimension of $\D_I$.
Finally, the $k-s$ first rows of $\G$ give a generator matrix of the shortened GRS code, which is itself a GRS code. This yields the result on the structure of $\D_I$.
\end{proof}

\begin{proof}[Proof of Proposition~\ref{prop:shorten_qGRS}]
By definition, $\C$ has codimension $\ell$ in a code \[\D \eqdef \GRS_k(\aav, \vv) \oplus \C_1\] for some code $\C_1$ of dimension $\ell$. By Lemma~\ref{lem:short_superGRS}, $\D_I$ has dimension $k + \ell - |I|$. Since by assumption, $\dim \C_I = \dim \C - |I|$, then $\C_I$ has codimension $\ell$ in $\D_I$.
By Lemma~\ref{lem:short_superGRS} again, $\D_I = \GRS_{k-s}(\aav',\vv') \oplus \C_1'$ for some $\aav', \vv' \in \Fq^{n-s}$ and some code $\C_1' \subseteq \Fq^{n-s}$. Let $\C_0' \eqdef \C_I \cap \GRS_{k-s}(\aav',\vv')$. The code $\C'_0$ 
has codimension $\ell' \leq \ell$ in $\GRS_{k-s}(\aav', \vv')$ and since $\dim \C_I = k-s$, we deduce that $\C_I$ is the direct sum of a codimension $\ell'$ subspace of $\GRS_{k-s}(\aav',\vv')$ and a code of dimension $\ell$'. Hence it is an $\ell'$-qGRS code.
\end{proof}

\begin{remark}
Note that the typical scenario is that the shortening of an $\ell$-qGRS code
remains an $\ell$--qGRS. Cases where $\ell$ decreases are sporadic.
For $\ell =1$, we can observe the following.
\end{remark}
\begin{lemma}\label{codim1 lemma}
Let $\EE$ be a subspace of codimension $1$ of $\Mon k$ and assume that 
\[p(x)= \prod\limits_{\alpha_i,i  \in I}(x-\alpha_i)\quad \text{for}\quad I \subsetneq [n]\quad
\text{with}\quad |I| = a.\] Then $\EE_{p(x)}$ has codimension $0 $ or $1$ in $p(x) \Mon{k-a} $.
    \end{lemma}
\begin{proof}
By definition, $\EE_{p(x)} \subseteq p(x) \Mon{k-a}$ and $\dim p(x) \Mon{k-a} = k -a$. Thus $\dim \EE_{p(x)} \le k -a$. 
On the other hand, note that \begin{equation}
\EE_{p(x)} = \bigcap_{i=1}^{a} \ker  \left( \begin{tabular}{l}
  $\EE \to \Fq$\\
  $f \mapsto f(\alpha_i)$
\end{tabular}       \right).
\end{equation}
This implies $\dim(\EE_{p(x)}   )   \ge \dim(\EE ) -a =k-1-a$, and the result follows.
\end{proof}

\begin{remark}
It turns out that, in general, the shortening of a TGRS code is not a TGRS one.
Therefore, Proposition~\ref{prop:shorten_qGRS} highlights an important interest of the notion of qGRS codes compared to TGRS codes. This will be particularly useful
in the sequel, as the shortening operation plays a crucial role in the distinguishers and attacks to be discussed.
\end{remark}



\section{TGRS codes and a distinguisher based on Schur product}\label{sec:distinguisher}

In this section, we show how shortening and squaring can be exploited to distinguish
qGRS codes from random ones for suitable parameters. From Proposition~\ref{prop:shorten_qGRS}, this
straightforwardly leads a distinguisher on TGRS codes for suitable parameters.
We start by discussing the general setting, \emph{i.e.} an arbitrary number of $\ell$ twists. We then focus specifically on the case where $\ell=1$, for which we will derive an attack in the subsequent section.

\subsection{A distinguisher for general $\ell$}

\subsubsection{Polynomial setting.}\label{ss:poly_setting}
Let  $(\aav, \vv)$ such that $\aav$ is a sequence of distinct elements and $\vv$ is a sequence of nonzero elements both in $\Fq$. We have the corresponding evaluation map
$\evav : \Fq[x] \rightarrow \Fq^n$.

In what follows, we consider a qGRS code $\C$ (see Definition~\ref{def:qGRS}) such that
\[
    \C = \C_0 \oplus \C_1 
\]
where $\C_0$ has codimension $\ell$ in $\GRS_k(\aav,\vv)$ (thus $\C_0$ is a subcode of $\GRS_k(\aav,\vv)$), $\C_1$ has dimension $\ell$ and $\C_1 \cap \GRS_k(\aav,\vv) = 0$.
In particular, $\C$ has dimension $k$.

Via the correspondence discussed in \S~\ref{ss:poly}, $\C = \evav (\mathcal{P})$
such that $\P \subseteq \Fq[x]_{<n}$ and
\[
    \P = \P_0 \oplus \P_1,
\]
where $\P_0$ has codimension $\ell$ in $\Fq[x]_{<k}$ and $\P_1 \cap \Fq[x]_{<k} = 0$.
Denote $v_1, \dots, v_\ell$ a basis of $\P_1$ whose elements have strictly
increasing degrees and denote by
\begin{equation}\label{eq:dl}
d_\ell \eqdef \deg (v_\ell)
\end{equation}
which is the maximal possible degree for an element of $\P$.

\subsubsection{Inequalities for distinguishing.}
For distinguishing $\C$ from a random code, we examine the dimension of $\C^2$
or $\C_I^2$ for some $I \subseteq [n]$.
For $\C$ to be distinguishable from random codes, we need
\[
    \dim \C_I^2 \leq \min\left\{n - |I|, {\dim \C_I + 1 \choose 2}\right\}.
\]
\begin{remark}
For $|I| < k$, the typical situation is that $\dim \C_I = k - |I|$
and we will assume this situation throughout this section. Note that
if this condition does not hold for almost any such $I$, such an observation would yield a distinguisher on the code.
\end{remark}

\begin{proposition}\label{prop:dim_sq_qGRS} Let $\C$ be an $\ell$-qGRS code
  of dimension $k$ obtained by evaluation of $\P \subseteq \Fqx$ with
  maximal degree $d_\ell$. Then,
\[
    \dim \C^2  \leq \min \left\{ (\ell + 2)k-1 -  \frac{\ell(\ell-1)}{2} ,\ k+ d_\ell+ \frac{\ell(\ell+1)}{2} ,\ 2 d_\ell + 1  \right\}.
\]    
\end{proposition}
    
\begin{proof}
Consider the notation introduced in \S~\ref{ss:poly_setting}.
 Proposition~\ref{prop:bounding_dim_prod} asserts that
\[\dim \C^2 \leq \dim \P^2.\]
Thus we will estimate $\dim {\P}^2$. Note that
\begin{equation}\label{eq:P2}
    \P^2 = \P_0^2 + \P_0 \P_1 + \P_1^2.
\end{equation}
Since $\P_0 \subseteq \Fq[x]_{<k}$,
Proposition~\ref{prop:prod_space_poly_bounded} asserts that
$\dim \P_0^2 \leq 2k-1$. Then, by
Proposition~\ref{prop:dim_prod_space_poly},
\[
    \dim \P_0 \P_1 \leq \ell (k - \ell) \quad \text{and} \quad
    \dim \P_1^2 \leq \frac{\ell(\ell + 1)}{2}\cdot
\]
Including the last estimates in \eqref{eq:P2} yields:
\[
    \dim \P^2 \leq 2k-1 + \ell (k - \ell) + \frac{\ell (\ell+1)}{2} = (\ell+2)k-1 - \frac{\ell (\ell-1)}{2},
\]
from which we can deduce the first upper bound.

Next observe that, from the definition of $d_{\ell}$, $\P \subseteq \Fq[x]_{<d_\ell + 1}$. Therefore, rewriting \eqref{eq:P2} as
\[
    \P^2 \subseteq \P_0 \Fq[x]_{<d_\ell +1} + \P_1^2,
\]
then, since $\P_0 \subseteq \Fq[x]_{<k}$, from Proposition~\ref{prop:prod_space_poly_bounded} we get,
\[
\P^2 \subseteq \Fq[x]_{< k + d_{\ell}} + \P_1^2
\]
which yields the second upper bound:
\[
    \dim \C^2 \leq k + d_\ell + \frac{\ell (\ell + 1)}{2}\cdot
\]
Finally, from the inclusion $\P \subseteq \Fq[x]_{<d_\ell +1}$
together with Proposition~\ref{prop:prod_space_poly_bounded}, we can
deduce that $\P^2 \subseteq \Fq[x]_{<2d_\ell +1}$ yielding the third
inequality:
\[
    \dim \C^2 \leq 2 d_\ell +1.
\]
\end{proof}    
    
\begin{remark}
  The third upper bound $\dim \C^2 \leq 2 d_\ell +1$ corresponds to
  the dimension of the square of the code referred to as the
  \emph{outer code} in \cite{BBPR}.
\end{remark}

\begin{theorem}\label{thm:dimCI2}
    Let $\C$ be an $\ell$--qGRS of dimension $k$ and $I \subseteq [n]$
    be a set such that $|I| < k$ such that $\dim \C_I = k - |I|$, then,
    \begin{align*}
        \dim \C_I^2 \leq \min \Big\{
            (\ell+2)(k-|I|) &- 1 - \frac{\ell (\ell - 1)}{2},\\
            k& + d_\ell - 2|I| + \frac{\ell (\ell+1)}{2},\
            2(d_\ell - |I|) + 1
        \Big\}.
    \end{align*}
\end{theorem}

\begin{proof}
According to Proposition~\ref{prop:shorten_qGRS}, the shortening of $\C$
is still an $\ell'$--qGRS code for $\ell' \leq \ell$. Thus, the result is
a direct consequence of Proposition~\ref{prop:dim_sq_qGRS} where the following changes of variables, are applied:
\[
    \ell \leftarrow \ell' \quad
    k \leftarrow k-|I| \quad
    d_{\ell} \leftarrow d_\ell - |I|.
\]
Two things may require clarifications.
First, according to this change of variables the result should involve $\ell'$ and not $\ell$. However, since the upper bounds of Proposition~\ref{prop:dim_sq_qGRS} are increasing functions in $\ell$, if they hold for $\ell' \leq \ell$ they still hold for $\ell$. Thus $\ell'$ can be replaced by $\ell$ in the final bounds.

Second, the change of variable $d_\ell \leftarrow d_\ell - |I|$ requires further details. From \S~\ref{sss:shorten_poly}, the shortened code $\C_I$ corresponds to the
polynomial space $\P_{p_I}$ where $p_I(x) \eqdef \prod_{i \in I} (x-\alpha_i)$.
Since $\P \subseteq \Fq[x]_{< d_\ell + 1}$, then $\P_{p_I} \subseteq p_I \Fq[x]_{< d_\ell - |I| + 1}$. Therefore, when shortening by
forgetting the common factor $p_I$, we effectively deal with polynomials of degree $\leq d_\ell - |I|$.
This concludes the proof.
\end{proof}

\begin{corollary}\label{cor:sq_TGRS}
Let $\C$ be a $(\bm{\a},\bm{v},\mathbf{t},\mathbf{h},\mathbf{\eta})$ TGRS code of dimension $k$.
Denote by $t_{max}$ the largest entry of $\mathbf{t}$. Let $I \subseteq [n]$ with $|I| < k$. Then
\begin{align*}
        \dim \C_I^2 \leq \min \Big\{
            (\ell+2)(k-|I|) &- 1 - \frac{\ell (\ell -1)}{2},\\
            2k + t_{max}-1 - 2|I| &+ \frac{\ell (\ell+1)}{2},\
            2(k+t_{max} - |I|) - 1
        \Big\}.
    \end{align*}
\end{corollary}

\begin{proof}
This is a straightforward consequence of  Theorem~\ref{thm:dimCI2} after observing that $d_\ell = k + t_{max}-1$.
\end{proof}

\subsubsection{The range of the distinguisher.}
There remains to analyze for which parameters qGRS codes can be distinguished from random ones. Since the class of qGRS codes is closed under duality,
one can assume that $k \leq \frac n 2$ to estimate the range of the
distinguisher.

Now our objective is to identify for which parameters $(n, k, \ell)$ there
exists $a < k$ such that for $I \subseteq [n]$ with $|I|=a$ we have
\begin{equation}{\label{eq:dist_short}}
    \dim \C_I^2 < \min \left\{n-a, \frac{(k-a)(k-a+1)}{2} \right\}.
\end{equation}
Theorem~\ref{thm:dimCI2} yields 3 distinct bounds on $\dim \C_I^2$. The first one
is that which decreases the fastest in $a$, thus, we will only consider this one,
which permits a simpler analysis. Later on, when focusing on $\ell = 1$, the other bounds will be useful.

Now, we have to prove the existence of $0 \leq a < k$ such that
\[
    (\ell + 2)(k-a) - 1 - \frac{\ell (\ell-1)}{2} < \min \left\{
    n-a, \frac{(k-a)(k-a+1)}{2}\right\}.
\]
A calculation proves that the upper bound $< n-a$ holds for any $a$
satisfying
\begin{equation}\label{eq:lower_a}
    a > \left(\frac{\ell + 2}{\ell +1}\right) k {- \frac{\ell(\ell-1)}{2(\ell +1)} - \frac{n}{\ell+1}}\cdot
\end{equation}
For the second inequality, we consider a stronger and simpler one:
\[
    (\ell+2)(k-a) \leq \frac{(k-a)^2}{2},
\]
which yields 
\begin{equation}\label{eq:simple_a}
    a \leq k - 2(\ell+2).
\end{equation}
Now, putting \eqref{eq:lower_a} and \eqref{eq:simple_a} together,
we can conclude that if there exists an $a$ where $a=|I|$ and such that
 $$\left(\frac{\ell + 2}{\ell +1}\right) k {- \frac{\ell(\ell-1)}{2(\ell +1)} - \frac{n}{\ell+1}}<a\leq  k - 2(\ell+2)$$
 then dimension of the shortened square code $\CC_I^2$ satisfies \eqref{eq:dist_short} and so we can distinguishable these kind of codes from random ones.

\begin{theorem}
    An $\ell$--qGRS code of dimension $k$ is distinguishable from random as soon as:
    \[
    {n - k \geq \frac{3}{2} \ell^2 + \frac{5}{2} \ell + 4}.
\]
\end{theorem}
Recall here that we supposed $k \leq \frac n 2$, the left--hand side is then
larger than $\frac n 2$ while the right--hand side is in $O(\ell^2)$.
For TGRS codes proposed for McEliece encryption, $\ell$ is supposed to be $O(1)$ since the decoding is exponential
in $\ell$. So we can conclude that: \textbf{TGRS codes proposed for McEliece encryption are distinguishable from random}.

\subsection{The case $\ell=1$}
Let us apply the previous results to a single twisted TGRS code, \emph{i.e.}
in the case $\ell = 1$.
Recall that $\C=\ev(\P)$ where $\P=\Monh{k}{h}+\Span{x^h+\eta x^{k-1+t}}$.
The previous analysis provides.  
    
\begin{lemma}\label{3k bound}\label{2k+t bound}
Let $\C$ be an $(\bm{\a},\bm{v},\mathbf{t},\mathbf{h},\mathbf{\eta})$ $1$--TGRS code of dimension $k$ with $\mathbf{t}=(t)$ for some positive integer $t$.
Let $I \subseteq [n]$ such that $|I| < k$ and $\dim \C_I = k - |I|$.
Then,
\[
\dim \C^2 \leqslant
\begin{cases}
    3k-1 & \text{ for } t \ge k, \\
    2k+t & \text{ for } t < k.
\end{cases},
 \quad
\dim \C^2_{I} \leqslant \begin{cases}
    3(k -|I|)-1 & \text{ for } t \ge k - |I|, \\
    2(k - |I|) +t & \text{ for } t < k - |I|.
\end{cases}
\]

\end{lemma}

\begin{proof}
From Corollary~\ref{cor:sq_TGRS}, applied with $\ell = 1$ and $t_{max}= t$, we get
\[
    \dim \C_I^2 \leq \min \left\{3(k-|I|)-1,\ 2k+t-2|I| \right\}.
\]
This yields the proof.
\end{proof}

Next, we can redo the previous analysis in this case while making fewer approximations. We prove that the shortened code $\C_I$ has a square whose dimension differs from the typical one as soon as
\begin{equation}\label{eq:dist_range}\frac{3k-n}{2} \leq |I| \leq k-5.
\end{equation}
Thus, such a code is distinguishable from random as soon as
\[n-k {\geq 10}.\]
Since we supposed $k\leq  \frac n 2$, any $1$--qGRS code of length $>20$ is distinguishable from random.


\section{ A key-recovery attack on McEliece scheme with TGRS codes using Schur squares}\label{sec:5}
In this section we present a key recovery attack for the McEliece cryptosystem instantiated with a TGRS code (see Section~\ref{sec:McEliece}). Following the previous notations, let $\CC$ be the public code generated by $\Gpub$, a generator matrix of a $q$-ary single-twisted TGRS code $\mathcal{C}_{\bm{\a}, \bm{v}, \mathbf{t},\mathbf{h},\bm{\eta}} ^{n,k}  \subset \Fq^n$ of dimension $k$. 

Note that, unlike \cite{LR}, 
we do not impose any restriction on the parameters of the secret TGRS code. 


\subsection{Context and notation}

Let $\CC$ or equivalently, the $q$-ary single-twisted TGRS code
$\mathcal{C}_{\bm{\a}, \bm{v}, \mathbf{t},\mathbf{h},\bm{\eta}} ^{n,k}
$ have dimension $k \le \frac{n}{2}$ ({otherwise, one can work with
  the dual code of $\CC$}).

\begin{notation}
As $\CC$ has a single twist, we let $t, h$, and $\eta$ to be the twist, the hook, and the coefficient of $\CC$, respectively.
Let $\P = \MM  \oplus \Span{ x^h + \eta x^{k-1+t}}$, where $\MM = \Monh{k}{h}$. For a subset $I \subseteq [n]$, we define $p_I(x) = \prod\limits_{i\in I}(x - \a_i)$. 
\end{notation}
Recall that the injectivity of $\evav$ when restricted to $\Fqx_{<n}$
implies $ \CC = \evav(\P) $ and the shortening of $\CC$ at $I$ is
$\CC_I = \evav(\P_{p_I(x)})$. 
{Note that, $\MM$ is a
codimension 1 subspace of $\Fqxk$. By Lemma \ref{codim1 lemma}, $\MM_{p(x)}$ is a subspace of $p(x)\Mon {k-a}$ of codimension at most 1.}
 Here is an inclusion diagram for the involved polynomial spaces, which are isomorphic to involved codes via the map $\evav$.


\begin{figure}[H]
    \centering
    \includegraphics{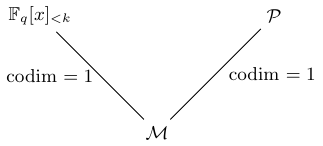}
    \quad\quad
    \includegraphics{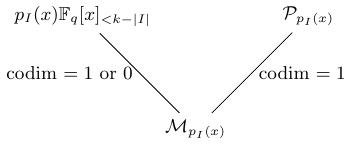}
\end{figure}
\subsection{Key-recovery algorithm}

We brief the main steps of our key recovery algorithm, which are given in detail in the following sections.
\begin{enumerate}
    \item[Step 1] Compute a basis of $\evav(\MM)$ or $\evav(\MM_{p_I(x)})$ using a distinguisher-based method (Section~\ref{keyrec:step1}).
    
    \item[Step 2] Use a Sidelnikov-Shestakov-like method given in \cite[\S~6.3]{CGGOT} to recover the secret evaluation vectors and multipliers for $\CC$ (Section~\ref{keyrec:step2}).
    
    \item[Step 3] Compute the other secret parameters which are the twist $t$, the hook $h$, and the coefficient $\eta$ (Section~\ref{keyrec:step3}).
    
\end{enumerate}
\subsubsection{Step 1: recovery of a codimension one subcode of the public (shortened) code.}\label{keyrec:step1}

\noindent

Following Lemma~\ref{3k bound}, we know that $\CC$
can be distinguished by considering its Schur square for the case $3k < n$ and otherwise by taking the Schur square of a shortened code $\CC_I$ with $(3k -n)/2< |I| < k -5$ (according to \eqref{eq:dist_range}). Now, depending on $t$, we consider the following two cases. 
\begin{enumerate}
    \item \label{case1} If $2k -2 +t < n$, then we compute a basis of the subcode $\evav(\MM)$ of $\CC$. 
    \item \label{case2} Otherwise, we compute a basis of the subcode  $\evav(\MM_{p_I(x)})$ of the shortened code $\CC_I$ for $I$ such that $2k -2 +t -n < |I|$.
\end{enumerate} 


   Note that in Sec. \ref{pol}, we work with polynomial spaces and to apply those results to the corresponding evaluation codes, we need $\evav$ to be injective. Therefore, the choice of parameters in Case~\ref{case1} has been made so that $\evav$ is injective on the product space $\MM \P$. In Case~\ref{case2}, to make sure the injectivity of $\evav$ on $\MM_{p_I(x)} \P_{p_I(x)}$, $I$ should be chosen such that 
   $2(k-|I|) -2 +t < n -|I|$, or equivalently, $2k + t -2 -n < |I|$. Injectivity of $\evav$ is crucial to apply the results on polynomial spaces, for instance, Proposition \ref{algorithm1} and Theorem \ref{main thm*} to the corresponding codes.
 
   First, we observe that, in Case~\ref{case1}, for a basis $\evav(e_1), \ldots, \evav(e_k)$ of $\CC$, if we take three random elements $\evav(f_1), \evav(f_2), \evav(f_3)$ from its monomial subcode $\evav(\MM)$, then the linear space spanned by
   \[\{\evav(f_i) \schur \evav(e_j): i = 1,2,3,\  1 \le j \le k\}\] has low dimension, more precisely bounded by $2k+2$. This is due to the following statement.
    \begin{proposition}\label{algorithm1}
Let $f_1,f_2,f_3  \in \MM$, then
\begin{align}
    \dim ( \Span{f_1,f_2,f_3} \P) \leq 2k+2.\label{algorithm1_inequality}
\end{align}
\end{proposition}
\begin{proof}
  The result follows from the fact that $\MM \subset \Fqxk$ and 
  $\P \subseteq \MM \oplus \Span{c}$ for $c(x) = x^h + \eta x^{k-1+t}$.
  Thus,
  \[
    \Span{f_1,f_2,f_3}\P \subseteq \Span{f_1,f_2,f_3} \MM +
    \Span{f_1 c,f_2 c,f_3 c}.
  \]
  Since $f_1,f_2,f_3 \in \MM$ and $\MM \subseteq \Fq[x]_{<k}$, the first term of the right--hand side is contained in $\Fq[x]_{<2k-1}$, which yields the result.
\end{proof}
It is important to note that Proposition~\ref{algorithm1} {also holds in Case~\ref{case2}} if we replace $\MM$ by $\MM_{p_I(x)}$ and $\P$ by $\P_{p_I(x)}$. More importantly, if we consider $f_1, f_2, f_3$ randomly from $\P$, then due to the higher degree term $x^{k-1+t}$, the space
$\Span{f_1,f_2,f_3}\P$
is very likely to have a much larger dimension than $2k+2$. This is indeed the case, as we prove the following result in Section~\ref{pol}.
Note that in Section~\ref{pol}, we consider the case where $\MM_{p(x)}$ has codimension 1 in  $p(x)\Mon {k-a}$. For the case $\MM_{p(x)}=p(x)\Mon {k-a}$, the probabilistic estimates actually become even  more favorable.

\begin{theorem}\label{main thm*}
Let $n,k,t$ be integers such that $17\leqslant k < n$ and $17 \leqslant t \le n-k$. Suppose $\R = \E \oplus \Span{c(x)}   \subseteq \Fq[x]$, where $\E \underset{\text{codim }1}{\subseteq} \Fqxk$ and $c(x)$ has degree $k -1 +t$.
    If $f_1,f_2,f_3\in \R$ and satisfy
    $$\dim( \Span{f_1,f_2,f_3}\R )\leqslant 2k+2,$$ then
    $f_1,f_2,f_3\in \E$ with probability $\geqslant 1-\frac{1}{q^2}$. Furthermore, a basis of $\E$ is recovered with probability $ \geqslant (1 - \frac{(k-3)}{q^2})$.
\end{theorem}
    
\begin{remark} Let us explain why exactly three functions are chosen in Theorem~\ref{main thm*} and in Algorithm~\ref{algorithm}.
If we draw $s$ elements of the public at random code, the probability that each one lies in some fixed codimension one subcode is $q^{-s}$. Hence, in average, we will have $O(q^s)$ trials to perform before succeeding. Thus, the larger the $s$ the larger the complexity of the attack.
Next, observe that with $s = 2$ we cannot distinguish elements of the subcodes since for any code
$\mathcal{D}$
\[\dim \langle f_1,f_2\rangle*\mathcal{D} \leq 2 \dim \mathcal{D} - 1 < 2 \dim \mathcal{D} +2.\] Thus $s=3$ is the least number of elements that yields a successful attack. Because of the aforementioned complexity issues there is no interest in considering a higher $s$.
\end{remark}
  Now, we apply Theorem~\ref{main thm*} by taking $\R = \P$ and $\E=\MM$ (Case \ref{case1}) or $\R=\P_{p_I(x)}$ and $\E=\MM_{p_I(x)}$ (Case \ref{case2}) to ensure the isomorphism between the polynomial spaces and the corresponding codes via the evaluation map $\evav$. 
  Based on the above theorem, which is proved in Section~\ref{pol}, we recover the codimension 1 subcode $\evav(\MM)$ of $\CC$ (in Case~\ref{case1}) or $\evav(\MM_{p(x)})$ of $\CC_I$ (in Case~\ref{case2}) with probability at least $1 - \frac{k-3}{q^2}$ using a distinguisher based method explained in Algorithm~\ref{algo:Clambdaperp}. This algorithm was introduced in \cite[Algorithm 1]{CGGOT} in a quite different version: we adapted it to the framework of our attack. Note that the parameter $k$ in Algorithm \ref{algo:Clambdaperp} refers to the dimension of $\R$. Therefore, when we apply the algorithm, it will be for $k$ in Case \ref{case1} and for $k-|I|$ in Case \ref{case2}.

  The recovery algorithm (Algorithm~\ref{algo:Clambdaperp}) proceeds as follows.
  
  \begin{description}
   \item[\textbf{In Case \ref{case1}:}]
     \begin{enumerate}
         \item Draw random triples $(\cv_1,\cv_2,\cv_3) \in \CC^3$ until we find one such that
         \[
         \dim (\Span{\cv_1, \cv_2, \cv_3} \CC) \leq 2k+2.
         \]
         According to Theorem~\ref{main thm*} such a triple is likely to be in $\evav(\MM)$.
         Since $\MM$ has codimension $1$ in $\CC$, the probaility of finding such a triple
         is $\frac{1}{q^3}$ and hence in $O (q^3)$ trials, we should find such a triple.
         \item  We collect other elements of $\evav(\MM)$ until we get a basis of this monomial subcode. For this sake, 
         draw $\cv \in \CC$ and keep it if
         \[
         \dim (\Span{\cv_1, \cv_2, \cv} \CC) \leq 2k+2.
         \]
         Such a $\cv$ is found in $O(q)$ trials.
     \end{enumerate}
\item[\textbf{In Case \ref{case2}:}] We do the same while replacing $\CC$ by $\CC_I$ and $\MM$ by $\MM_{p_I}$. Next, we iterate this process by choosing subsets $I_1, \ldots, I_s \subset [n]$ satisfying  $\bigcap\limits_{i=1}^s I_i = \emptyset$ and $(3k-n)/2 < |I_i| < k-5$ for all $i = 1, \ldots, s$. The last condition being taken to recover the codimension 1 subcodes $\evav(\MM_{p_{I_i}})$ of $\CC_{I_i}$.
     We collect shortenings $\evav (\MM_{p_{I_i}})$ and sum them up until the sum has dimension $k-1$. This permits to recover the monomial subcode $\evav(\MM)$.
     
     This recovery of a subcode by summing up shortenings is already used 
     and discussed in \cite[\S~IV.F]{COT17}.
\end{description}
In summary, whatever the parameters $k,t$, in the end of this step we have a basis of
$\evav(\MM)$.

\begin{algorithm}[t] 
  {\bf Input  : } A basis $\{\cv_1,\ldots ,\cv_{k}\}$ of $\C$ \\
  {\bf Output : } A basis $\mathcal{B}$ of $\C_0$.

  \begin{algorithmic}[1]
  \REPEAT \label{re}
   \FOR{$1 \leqslant  i \leqslant  3$}
   \STATE{Randomly choose $\bv_i$ in $\C$}
   \ENDFOR
   \STATE{ $\mathscr{D} \leftarrow ~ \langle \big\{ \bv_i \schur \cv_j ~|~ 1 \leqslant i \leqslant 3 \text{~and~} 1 \leqslant j
  \leqslant k \big \} \rangle $}
  \UNTIL{$\dim(\mathscr{D}) \leqslant  2k+2$ and $\dim ( \langle \bv_1, \bv_2, \bv_3 \rangle ) =3 $}
  \STATE{$\mathcal{B} \leftarrow \{\bv_1,\bv_2,\bv_3\}$}
  \STATE{$s \leftarrow 4$}
  \WHILE{$s \leqslant  k-1$}
    \REPEAT
  \STATE{Randomly choose $\bv_s$ in $\C$}   
  \STATE{$\mathcal{E} \leftarrow ~ \langle \big\{ \bv_i  \schur \cv_j ~|~ i \in \{1,2,s\} \text{~and~} 1 \leqslant j
  \leqslant k \big \} \rangle $}
  \UNTIL{$\dim(\mathcal{E}) \leqslant  2k+2$ \AND $\dim \left(\langle \mathcal{B} \cup \left\{
      \bv_s \right \} \rangle \right) = s $}
  \STATE{$\mathcal{B} \leftarrow \mathcal{B} \cup \{\bv_s\}$}
  \STATE{$s \leftarrow s+1$}
  \ENDWHILE
    \RETURN{$\mathcal{B}$;}
  \end{algorithmic}
  \caption{\label{algo:Clambdaperp}Recovering $\C_0$.}
    \label{algorithm}
\end{algorithm}

\subsubsection{Step 2: recovery of the secret evaluation vector and the multipliers.} \label{keyrec:step2}
Now we have access to the space $\evav(\MM)$. Since $\MM \subseteq \Fq[x]_{<k}$ and
we assumed $k \leq \frac n 2$, we see that $\MM^2 \subseteq \Fq[x]_{<2k-1}$ and hence
$\evav (\MM) \varsubsetneq \Fq^n$. Moreover, $\MM  = \Span{1,x, \dots, \widehat{x^h}, \dots, x^{k-1}}$. A classical argument from combinatorics permits to prove that if $h \neq 1$
or $k-2$, then $\MM^2 = \Fq[x]_{<2k-1}$. In such a situation, $\evav(\MM)^2$ is a GRS
code, and then the Sidelnikov--Shestakov attack can be applied to it as it is done in \cite{Wie10}.

This permits to recover a candidate for the pair $(\aav, \vv)$.

\begin{remark}
We did not consider here the pathological cases $h = 1$ or $k-2$. They should be subject to a separate development that we do not treat in this article.
\end{remark}

    
   
   \subsubsection{Step 3: recovery of the hook $h$, the twist $t$, and the coefficient $\eta$ of the secret key.} \label{keyrec:step3}
   The previous steps provide $\bm{\a}$ and $\bm{v}$ which we use to recover the hook $h$, the twist $t$ and the coefficient $\eta$. 
   \begin{description}
\item[\textbf{Recovering $\bm{h}$.}] For $0 \le i \le k-1$, we check if $\bm{v} \schur \bm{\a}^i \in \CC$ or not. Our assumption on $\G_{pub}$ implies that there is exactly one $i$ for which $\bm{v} \schur \bm{\a}^i \notin \CC$, and that integer would be the hook $h$.
   
\item[\textbf{Recovering $\bm{t}$.}]
    Take a codeword $\bm{c}\in \CC \setminus \evav(\MM)$. By interpolation we get $f \in \Fq[x]_{<n}$ such that $c = \evav(f)$ and $\deg(f) = k-1+t$ for some positive $t$.
    Which yields $t$.
    
\item[\textbf{Recovering $\bm{\eta}$.}]
Still considering the same $f$ and denoting by $f_i$ its coefficient of degree $i$
then $\eta$ satisfies $f_{k-1+t} = \eta f_h$, which permits to deduce $\eta$.
   \end{description}

\subsubsection{Attacked parameters (with a proof)}
We provide the sets of parameters in Table~\ref{tab:side}, which includes the TRS codes considered in \cite{LR} for the single twist case (labeled as [LR]) as well as the parameters for the provable attack discussed in this article (labeled as [CPTZ]).

We emphasize that,
\begin{itemize}

\item In the table, we consider only parameters for which Theorem~\ref{main thm*} holds. Note that, in Case \ref{case2} (as mentioned in the beginning of Subsection~\ref{keyrec:step1}), we can apply Theorem \ref{main thm*} if $k - |I| < 17$. Thus considering the conditions $2k+t -2 -n < |I| \le k - 17$, we get $t < n-k-15$. 
  

\item \cite{LR} restricts to a limited subclass of TRS codes while we consider the much broader class of TGRS codes (the attack can be extended to qGRS codes with no difficulty).
\begin{itemize}
    \item $q = q_0^2$ and $k < n \le q_0 -1$ where $q_0$ is a prime power,
    \item $\frac{n}{3} + \sqrt{n}+\frac{1}{3} < k \le \frac{n}{2}-2$ (the left hand side inequality comes from the lower bound for $\ell$ in the proposed set of parameters), 
    \item $r = \lceil \frac{n+1}{3}\rceil + 2$,
    
    \item $t = 2 \lceil \frac{n+1}{3}\rceil - k +2$ and thus $t \ge \frac{2n}{3} - \frac{n}{2}+4 = \frac{n}{6} + 4$.
    
    \item $h = r$ which implies that $h \neq 1, k-2$.
\end{itemize}
\end{itemize}

\begin{table}[H]
\centering
    \begin{tabular}{l|cc}\toprule
         & [LR]\quad & [CPTZ] \\
        \midrule
        $q$ \quad & $n^2$ \qquad& \quad  $n$  \vspace{0.2cm} \\
        \vspace{0.2cm}
        $k\in $ \quad $ $ & $ $\quad $\left[\frac{n}{3} + \sqrt{n}+\frac{1}{3},  \frac{n}{2}-2\right]$\qquad & \quad$[\sqrt{2n}, n-14]$\\
        \vspace{0.2cm}
        $t$ \quad & $2 \lceil \frac{n+1}{3}\rceil - k +2$ \qquad&\quad $ [17, n-k -16]$\\
        \vspace{0.2cm}
        $h$ \quad &$\lceil \frac{n+1}{3}\rceil + 2$ \qquad& \quad$\neq 1, k-2$\\
        \bottomrule
    \end{tabular}
    \bigskip
    \caption{The parameters for provable attacks in the case of single twist}%
\label{tab:side}
\end{table}

We mention that in \cite{BBPR}, $t \ge \frac{2n}{3} - \frac{n}{2}+4 = \frac{n}{6} + 4$. Thus, $t < 17$ will imply $n <78$, which is really small compared to the practical parameters. 

\subsection{Complexity}
According to \cite[Prop.~2]{COTG15}, the complexity of computing a Schur square is $O(n^4)$. Next, the bottelneck of the attack is the first step with the recovery
of the first triple $(\cv_1, \cv_2, \cv_3) \in \evav({\MM})^3$.
 We claim that this step
is the dominant one in the running of the attack. Since, this step requires
$O(q^3)$ trials this yields an overall complexity of $O(q^3 n^4)$ operations in $\Fq$.

\subsection{Some comments about the non covered cases}\label{sec:pathology}

We observe the cases where the key-recovery attack might fail: 
\begin{itemize}
    \item if $h = 1 $ or $k-2$, the square of the codimension 1 subcode we recover in Step 1, might not have Schur square equal to $\GRS_{2k-1}$. 
    \item for $k,t < 17$, %
    Theorem~\ref{main thm*} is no longer valid.
\end{itemize}
Here we give some ideas on alternative ways to recover a valid secret key without going into detailed description.

\subsubsection{When $t$ is small.}
We consider in particular the case where $t < k$ and that we still assume that
$k < \frac n 2$. Recall that $\C = \evav({\P}$), where
$ \P = \MM \oplus \Span{c(x)}$ with
\[
\MM = \Span{1, x, \ldots, \widehat{x^h}, \ldots, x^{k-1}}
\]
for some $h \le k-1$ and $\deg(c)=k-1+t$.

Suppose first that $2k+t < n-2$. 
In this situation, observe that:
\[
\P^2 = \MM^2 + \c \MM + \Span{c^2}.
\]
Since $\deg(c) < k+t < 2k$ then $\MM^2 + \c \MM \subseteq \Fq[x]_{2k+t}$ and it is very likely that this inclusion is an equality. In this case, 
$\C^2$ is the direct sum of a GRS code and the one-dimensional space $\evav(\Span{c^2})$.
In this situation, taking the dual of $\C^2$ yields a codimension $1$ subcode of a GRS code whose structure can be recovered using Wieschebrink's attack \cite{Wie10}.

When $2k+t > n$, a similar approach can be applied to shortenings of the code.





\section{Implementation}\label{sec:implem}
The first part of the attack, that is to say the computation of the monomial subcode (Step 1, Subsection~\ref{keyrec:step1}) is implemented in the computer algebra system SageMath v9.5 \cite{sage}.
It is available on \href{https://github.com/nihantanisali/TGRS}{https://github.com/nihantanisali/TGRS}.
Since the rest of the attack consists in applying Sidelnikov Shestakov attack together with other classical routines we considered that checking the validity of Step 1 was sufficient.
Note that our implementation is a proof of concept that is far from being optimized. Hence, due to time and resource constraints we chose smaller $q$ and $k$. The results are summarized in Table \ref{table imp}. 
Even though these small-scale instances do not fully capture the computational challenges posed by larger, standard parameters, they demonstrate that our method works in practice.

\begin{table}[h!]
\centering
\scalebox{0.805}{\label{table imp}
\begin{tabular}{||p{0.8cm}|p{0.8cm}|p{0.8cm}|p{0.8cm}|p{0.8cm}|p{0.8cm}|p{1cm}||p{1.8cm}||}
\hline
$q_0$ & $n$ & $k$ & $t$ & $l$ & $h$ & $a=|I|$ & Runtime (s) \\
\hline
$53$ & $51$ & $17$ & $17$ & $1$ &   $8$ & $-$ & 240 \\
$53$ & $51$ & $19$ & $8$  & $1$ &   $8$ & $-$ & 261 \\
$61$ & $51$ & $25$ & $17$ & $1$ & $19$& $8$ &  3834 \\
\hline
\end{tabular}}
\medskip
\caption{Experimental results obtained by averaging several runtimes 
         of Algorithm~\ref{algorithm} on a 13th Gen Intel® Core™ i7-13800H × 20.}
\end{table}
\section{A result on the polynomial space}\label{pol}
The aim of this section is to provide a proof of Theorem~\ref{main thm*}.
We will consider a polynomial subspace
$$ \R = \E \oplus \Span{c(x)}   \subseteq \Fq[x],$$ 
where $ \E \underset{\text{codim }1}{\subseteq} \Fqxk$, and $c(x) \in \Fq[x]$ such that $\deg(c)=k-1+t $ with $t > 0$. 
We will determine the quantity 
\begin{align*}
    \dim \< f_1, f_2, f_3\> \R
\end{align*}
for different choices of $f_1,f_2,f_3 \in \R $ depending on admissible values of $t$ and $k$, which will be specified later. 

First, recall that Proposition~\ref{algorithm1} yields an upper bound on $\dim \< f_1,f_2,f_3 \>\R$ for $f_1,f_2,f_3\in \E$:
\[ \dim \< f_1, f_2, f_3\> \R\leqslant 2k+2 .\] 
Then, we study the general case of $f_1,f_2,f_3\in \R$
in order to prove Theorem \ref{main thm*}. First we define the following sets. 
\begin{align}
    \Psi&\eqdef\{(f_1,f_2,f_3)\in \R^3 : \dim \< f_1, f_2, f_3\> \R > 2k+2 \};\label{psi}\\
    \Gamma&\eqdef \{(f_1,f_2,f_3)\in \R^3\setminus\E^3 :\dim \< f_1, f_2, f_3\> \R \leqslant  2k+2 \}. \label{gamma}
\end{align}
We get a partition of $\R^3= \E^3 \sqcup \Gamma \sqcup \Psi$ as illustrated in Figure~\ref{Fig.1}. 
\begin{figure}[H]
    \centering
    \includegraphics{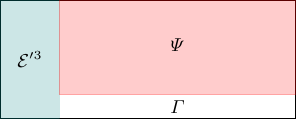}
%
     \caption{Illustrates the partition $ \E^3 \sqcup \Psi \sqcup \Gamma $ of $ \R^3 $}
     \label{Fig.1}
\end{figure}
Our main goal in this section is to prove that if $ (f_1,f_2,f_3) \in \R^3$ with $ \dim\< f_1, f_2, f_3\> \R \leqslant 2k+2$, then  $(f_1,f_2,f_3)\in \E^3$ with probability $\geqslant 1-\frac{1}{q^2}$. This probability corresponds to showing that
\begin{equation} \frac{|\E^3|}{| \E^3 \sqcup \Gamma|  }  \geqslant 1-\frac{1}{q^2}.\end{equation}

 Let $f_1, f_2, f_3 \in \R$. If $f_i \notin \E$ for some $1 \le i \le 3$, then w.l.o.g. (after taking linear combinations), we can assume that the polynomials are of the form
\begin{align*}
    f_1(x) &:= a_0+a_1 x+\cdots +a_{k-1}x^{k-1} + a_{k-1+t} c(x), \quad a_{k-1+t} \neq 0,\\
    f_2(x) &:= b_0+b_1 x+\cdots +b_{k-1}x^{k-1}, \\
    f_3(x) &:= c_0+c_1 x+\cdots +c_{k-1}x^{k-1}.
\end{align*}
First, note that
$
        \dim \< f_1, f_2, f_3\> \R \geqslant   \dim  \< f_1, f_2, f_3\> \E. 
$
We now focus on  $\dim  \< f_1, f_2, f_3\> \E$. 
Clearly, the dimension of the space $f_1 \E$ is $k-1$. We now fix a basis $\{e_1(x), \ldots , e_{k-1}(x)\}$ of $ \E$ such that $\deg(e_1)< \cdots < \deg(e_{k-1}).$

 For all polynomials $ e_i(x)\in \E $, we have that 
 \[k -1+t \leqslant \deg(e_i(x)f_1(x)) \leqslant 2k-2 +t.\]
 We define the subspace of $f_1\E$ generated by $e_i$'s of \textit{large degrees}, denoted by $(f_1\E)_l$, as follows:   
$$ (f_1\E)_l \eqdef \Span {f_1 e_ i ~\colon~ {\deg(f_1)+\deg(e_i) \ge 2k-1 } }.  $$
We have 
\begin{align}\dim(( f_1\E)_l ) =\begin{cases}
k-1 \quad  \text{ if $t \ge k $},\\
t-1 \text{   or  } t   \quad\text{ if $t < k $}.
\end{cases}  \label{t>k large polynomials}\
\end{align}
{For $t<k$ in the above \eqref{t>k large polynomials}, $\dim(( f_1\E)_l )$ has two possibilities depending on the missing integer $t_1 \in \{0,\dots, k-1\}\setminus \{\deg(e_i) ~\colon~ i=1, \ldots, k-1\}$. Indeed, if $k-1+t+t_1 \ge 2k -1$, then $\dim(( f_1\E)_l ) =t -1$ or $t$, otherwise. }
Using the fact that the space $ f_2  \E+ f_3 \E $ contains polynomials up to degree $2k-2$ (as $\deg(f_2),\deg(f_3) \leqslant k-1 $), we get the equalities
\begin{align}\label{combined}
    \dim  \< f_1, f_2, f_3\> \E =
    \left\{
    \begin{array}{lc}
       k-1 + \dim  \< f_2, f_3\> \E & \text{ if $t\ge k$,}\\
    \left. 
        \begin{array}{c}
          t-1 + \dim  \< f_2, f_3\> \E  \\ \text{ or, }\ \  t+ \dim  \< f_2, f_3\> \E  
        \end{array}
    \right\}     &  \text{if}\ t < k
    \end{array}
  \right.
\end{align}

We now focus on estimating $\dim  \< f_2, f_3\> \E$ where $f_2$ and $f_3$ are chosen uniformly at random from $\E$. 
For this, we use a classical result on polynomials that we prove for the sake of self containedness.
    
\begin{lemma}\label{sum}
     Let $f, g \in \Mon k$ be two polynomials.  
     Then 
    \[
    \dim(f \Mon k + g \Mon k ) = k+ \max\{\deg f,\deg g\} -\deg (\gcd(f,g)).
    \]
\end{lemma}
\begin{proof}
  Set $s = \deg f$, $u = \deg g$, $h = \gcd(f,g)$ and $d = \deg h$.
    W.l.o.g. one can assume $s=\max\{s,u\}$ and we set
     $f_1 (x)  = f(x) / h(x)$ and $g_1(x) = g(x) / h(x). $ 
     Consider the linear map
     \[
        \phi : \map{\Fq[x]_{<k} \times \Fq[x]_{<k}}{f\Fq[x]_{<k}+g\Fq[x]_{<k}}{(a,b)}{af+bg.}
     \]
     Any $(a,b) \in \ker \phi$ satisfies $af = -bg$. By dividing both sides by $h$
     we get $af_1 =- bg_1$ and since $f_1, g_1$ are coprime, we deduce that there exists $p \in \Fq[x]$ such that
     $a = g_1 p$ and $b = -f_1 p$.
     Since $a,b$ have degrees $<k$, we deduce that $\deg f_1 + \deg p = 
     s-d + \deg p < k$.
     Therefore,
     \[
        \ker \phi = \{(g_1 p, -f_1 p) ~:~ p \in \Fq[x]_{<k-s+d}\}.
     \]
     Thus, $\dim \ker \phi = k-s+d$ and the rank--nullity theorem permits to conclude.
     \end{proof}

Observe that $\dim \< f_2, f_3\> \E \ge \dim (f_2\Fqxk + f_3\Fqxk) - 2$. Now, \eqref{combined} implies that in order to get $\dim \<f_1, f_2, f_3\> \E > 2k+2$
it suffices to have
\begin{align}
    \dim\< f_2, f_3\> \E > \begin{cases}
    k+3 \quad \text{if } t \ge k,\\
    2k -t+3 \quad \text{if } t < k.
    \end{cases}
\end{align}
As a straightforward consequence of Lemma \ref{sum}, we reduce the dimension conditions to a condition on the degree of $\gcd(f_2,f_3)$. We give the condition more generally for any pair of polynomials $(f,g) \in \Mon k\times \Mon k$ and we call it the \emph{gcd condition} as defined below.
\begin{definition}\label{def:gcd_cond}
    We say that a polynomial pair $(f,g)\in \Mon k\times \Mon k $ with degrees $s $ and $u$  \emph{satisfies the gcd condition} if 
\begin{equation}\label{gcd_cond}
     \deg(\gcd(f,g)) < \begin{cases}
         \max\{s,u\} -5 & \text{ if $t>k$}\\
         \max\{s,u\} -(k-t) -5 &\text{ if $ t\leqslant k$.}
     \end{cases}
\end{equation}
\end{definition}


\begin{remark}\label{reduction}
 {Note that the bounds on the degree of the gcd in \eqref{gcd_cond} are taken so that, if a pair $(f_2,f_3) \in \E \times \E$ satisfies the gcd condition in Definition \ref{def:gcd_cond}, then $(f_1, f_2, f_3) \in \Psi$ for $f_1 \in \R \setminus \E$ (see \eqref{psi} for the definition of $\Psi$). Therefore, it helps to obtain a bound on the size of $\Psi$, that we determine in the subsequent part.}
\end{remark}
We aim to find the number of pairs $(f_2,f_3) \in \Mon k \times \Mon k$ that satisfy the gcd condition given in Definition \ref{def:gcd_cond}. For that, we introduce the required notations and notions.

\begin{definition}\label{newnotations}
Let  $i,j, s,u $ be non-negative integers such that $s,u <k$. Define
\begin{enumerate}
    \item $\mathcal{G} :=\{(f,g) \in \Fqxk \times \Fqxk : (f,g) \text{ satisfies the gcd condition} \}$;
    \item $\mathcal{G}_{\E}:=\{(f,g) \in \E\times \E :  (f,g) \text{ satisfies the gcd condition} \}$;
    \item $\mathcal{A}(s,u ):=\{( f,g) :\deg (f)= s ,~ \deg (g)=u \}$;
    \item $\mathcal{B}(s,u ,i):=\{( f,g)  :\deg (f)= s , ~\deg (g)=u,~ \deg(\gcd(f,g))  = i\}$;
    \item $  \mathcal{B}_j(s,u ):=\{( f,g)  :\deg (f)= s ,~\deg (g)=u ,~ \deg(\gcd(f,g)) \le j\}$.
\end{enumerate}
\end{definition} 

In the rest of this section, we will make use of the following result in our counting arguments.
\begin{theorem}[{\cite[Theorem 3]{BenBen07}}]\label{gcd theorem}
    Let  $f$ and $g$ be randomly chosen from the set of polynomials in $\F_q[x]$ of degree $ s$  and $u$ respectively, where $ s$  and $u$   are not both zero. Then the probability of $f$ and $g$ being coprime is $1-\frac{1}{q}$.
\end{theorem}
We start counting the number of polynomials in $\Fqx$ of fixed degrees.
\begin{lemma}\label{Counting lemma}\label{LEMMA}
    Let $s,u,i$ 
    be nonnegative integers such that $s,u < k$. Then $|\mathcal{A}(s,u )|= (q-1)^2 q^{s+u}$ and moreover, among the polynomials in $\mathcal{A}(s,u)$, 
    \begin{enumerate}
        \item $(q-1)^3 q^{s+u-1} $ many of them are coprime, i.e. $|\mathcal{B}(s,u,0)|= (q-1)^3 q^{s+u-1}$;
        \item $(q-1)^2 q^{s+u-1} $ many of them are not coprime;
        \item For $  i \le \min\{s,u\}$, $|\mathcal{B}(s,u,i)|=(q-1)^3 q^{s+u-i-1} $.
    \end{enumerate}
\end{lemma}

\begin{proof}
    The first two items follow directly from Theorem \ref{gcd theorem}.
    The number of monic polynomials $h$ of degree $i$ is $q^i$. Observe that the number of polynomial pairs $(f,g)$ with $ \gcd (f,g) =h$ is equal to the number of polynomial pairs $(f^{\prime} , g^{\prime }) $ with $\deg f^{\prime} = s-i$, $ \deg g^{\prime } =u-i$ and $\deg (\gcd(f^{\prime},g^{\prime}))=0 $. By the first item, there are $ (q-1)^3 q^{s+u-2i-1} $ such pairs. Since it is the same number for any $ \gcd(f,g) =h$, and we have $ q^i$ many choices for $h$, there are $ (q-1)^3 q^{s+u-2i-1} q^i =(q-1)^3 q^{s+u-i-1} $ many polynomial pairs with $\deg ( \gcd(f,g) ) =i $ where $i \le \min(s,u) $.
\end{proof}

To count the sets of polynomials $\mathcal{G}$ and $\mathcal{G}_{\E}$, we first count the number of polynomial pairs in $\mathcal{A}(s,u )$ having gcd less than or equal to $j$ where $j \le \min\{s,u\}$. 

\begin{lemma}\label{small gcd probability lemma}
For any positive integers $s,u$ and $ j \leqslant s,u< k$, we have
$$
\frac{  |   \mathcal{B}_j(s,u) |}{| \mathcal{A}(s,u) |} =  1- \frac{1}{q^{j+1}}.$$
\end{lemma}
\begin{proof} Note that $ \mathcal{B}(s,u,i ) $'s are non intersecting for distinct $i$ and thus $ \mathcal{B}_j(s,u) = \bigsqcup\limits_{i \leqslant j } \mathcal{B}(s,u,i )$. 
    Following Lemma \ref{LEMMA}, we have
    $ | \mathcal{B}(s,u,i )|  = (q-1)^3 q^{s+u-i -1} .$
    Hence,
    \begin{align*} \sum_{i\leqslant j } | \mathcal{B}(s,u,i )|  &= \sum_{ i\leqslant j } (q-1)^3 q^{s+u-i -1} \\& = (q-1)^3 q^{s+u-j-1} \left( 1+ \cdots+ q^{j} \right)   \\
    &= (q-1)^3  q^{s+u-j-1} \frac{( q^{j+1} -1 ) }{(q-1) }\\
    &= (q-1)^2   q^{s+u-j-1} ( q^{j+1} -1 ).
    \end{align*}
    Note that $| \mathcal{A}(s,u) | = (q-1)^2 q^{s+u}$, and therefore,
    \begin{align*}   
    \frac{  |  \mathcal{B}_j(s,u )  |}{| \mathcal{A}(s,u) |} =\frac{(q-1)^2   q^{s+u-j-1} ( q^{j+1} -1 ) }{(q-1)^2 q^{s+u }} = \frac{ q^{s+u }- q^{s+u-j-1} }{  q^{s+u}} =1-\frac{1}{q^{j+1}}. 
    \end{align*}
\end{proof}

Now, we are ready to count the polynomials in $\mathcal{G}$.
The lemma to follow yields a density estimate for $\mathcal G$ for $k,t$ large enough.
The lower bound on $k,t$ are chosen to further obtain a reasonable estimate for the density of $\mathcal{E}'$.

\begin{lemma}\label{main lemma}
For integers $k,t \ge 17$, we have
\begin{align*}
     \frac{| \mathcal{G} |}{|\Mon k \times \Mon k |}  \geqslant 1- \frac{1}{q^7}\cdot
\end{align*}
\end{lemma}
\begin{proof}
Note that
$  \mathcal{G}  =\bigcup_{s,u<k }  \mathcal{B}_j (s,u)$. Let $j_{s,u}$ be the greatest possible degree for $\gcd(f,g)$ such that $(f,g)$ satisfies the gcd condition where $\deg f =s$ and $\deg g = u$. Using this, we get
\begin{equation} \label{main lemma ineq 1}
    \frac{| \mathcal{G} |}{|\Mon k \times \Mon k |} = \sum_{s,u <k} \frac{| \mathcal{A} (s,u) |}{{|\Mon k |}^2 } \frac{ |\mathcal{B}_{j_{s,u}}(s,u) |}{ |\mathcal{A} (s,u)| }
\end{equation}
Observe that if we consider only the pairs $(s,u)$ such that $s+u=r \in \{ 2k-10,\ldots,2k-2\}$, we get the following lower bound of the sum in \eqref{main lemma ineq 1},
\begin{align}
    \sum_{s,u} \frac{| \mathcal{A} (s,u) |}{|\Mon k|^2 } \frac{ |\mathcal{B}_{j_{s,u}}(s,u) |}{ |\mathcal{A} (s,u)| } \geqslant
    \sum_{\substack{r  \\r \geqslant 2k-10 }} 
    \sum_{\substack{s,u\\ s+u=r} } \frac{| \mathcal{A} (s,u) |}{|\Mon k|^2 } \frac{ |\mathcal{B}_{j_{s,u}}(s,u) |}{ |\mathcal{A} (s,u)| }.\label{main lemma ineq 2} 
\end{align}
We replace $r $ by $r':=(2k-2)-r$ in the limits of the summation in the above equation \eqref{main lemma ineq 2}.
Note that for a fixed $r'$, there are exactly $( r' +1)$-choices for the pairs $ (s,u)$ such that $s+u=r$, and following the value of $| \mathcal{A} (s,u) |$ in Lemma \ref{LEMMA}, we have 
\begin{align}
    \frac{| \mathcal{A} (s,u) |}{|\Mon k|^2 } =\left( \frac{q-1}{q}\right)^2 \frac{1}{q^{r'}}. \label{main lemma ineq 3} 
\end{align} 
For $ r' =(2k-2)-r \leqslant 8$, and $s+u = r $, $\max\{s,u\} \geqslant k-5$.  For any pair $(s,u)$  such that $(2k-2)-r'= r = s+u \geqslant 2k-10$, this implies 
\[
\begin{cases}
&k-10 \leq j_{s,u} <    k-6 \quad  \text{ if } t > k,\\
& t-10 \leq j_{s,u} <   t-6 \quad \,\, \text{ if } t \le k.
\end{cases}
\]
Since  $k ,t  \ge 17$, then $j_{s,u} +1 \ge 8$ in both cases. Therefore, following Lemma \ref{small gcd probability lemma}, we get 
\begin{align}
     \frac{ |\mathcal{B}_{j_{s,u}}(s,u) |}{ |\mathcal{A} (s,u)| } = 1-\frac{1}{q^{j_{s,u}+1 }} \geqslant 1-\frac{1}{q^{8}}\label{main lemma ineq 4}.
\end{align}
Combining \eqref{main lemma ineq 1}, \eqref{main lemma ineq 2}, \eqref{main lemma ineq 3} and \eqref{main lemma ineq 4} all together, we get
\begin{align*}
    \frac{| \mathcal{G} |}{|{\Mon k}|^2} \geqslant  \left( \frac{q-1}{q}\right)^2 \left(   1-\frac{1}{q^{8}} \right) \sum_{r' =0}^{8}   (r'+1) \frac{1}{q^{r' }}  .
\end{align*}
The leftmost term is the truncated Taylor expansion of the function
 $F = (x-1)^{-2}$ evaluated at $1/q$. Let us write this the Taylor expansion as
$$  (x-1)^{-2}  =1+2x+3x^2 +\cdots + 8x^7 + 9x^8+R_8(x).  $$
Moreover, for any $x\in (0,1/5)$, we have 
\[
R_8(x)=\frac{F^{(9)}(z)x^{9}}{9!}
\]
 for some
$z \in (0, 1/5) $ and one can prove that for any such $x$, we have $  |R_8(x)| < 121 x^9.$
Taking $x=1/q$ (since we are dealing with GRS codes, one can reasonably assume that $q\geq 11$), we finally get
\begin{align*}
    \frac{| \mathcal{G} |}{|\Mon k|^2} &\geqslant \left( \frac{q-1}{q}\right)^2 \left(   1-\frac{1}{q^{8}} \right) 
    \left( \left(  \frac{q-1}{q}\right)^{-2} + O\left( \frac{1}{q^9}\right) \right)
   \geqslant 1-\frac{1}{q^7}.
\end{align*}

\end{proof}
\begin{remark}
    Note that in Lemma \ref{main lemma}, the lower bound $1-\frac{1}{q^7}$ can be improved for larger $ t,k$ such that $k,t\gg 17 $. As we can condider more $r$'s in  \eqref{main lemma ineq 2} for those cases, the error term in the Taylor expansion gets smaller.
    \end{remark}
The following Lemma will be useful to count polynomials in $\mathcal{G}_{\E}$.
\begin{lemma}\label{main lemma 2}
We have that, 
$$ \frac{| \mathcal{G}|}{|\Mon k \times \Mon k |} \geqslant 1- \frac{1}{q^w} \hspace{0.2cm} \Longrightarrow \hspace{0.2cm} \frac{| \mathcal{G}_{\E}|}{|\E\times \E|} \geqslant 1- \frac{1}{q^{w-2}}.$$
for some $w\in \Z_{>0}$.
\end{lemma}
\begin{proof}
Note that,
\begin{align*} 
    | (\E\times \E ) \cup \mathcal{G} |&=   | \E\times \E |+|\mathcal{G} | -| \mathcal{G} \cap (\E\times \E )  |\\
    & \leqslant |\Mon k   \times \Mon k| = q^2 |\E\times \E|. 
\end{align*}
Now, since $\mathcal{G}_{\E} = \mathcal{G} \cap (\E \times \E)$, from the previous inequality we can deduce that,
$$  |\mathcal{G}_{\E}|\geqslant (1-q^2) |\E\times \E| +|\mathcal{G}|. $$
From the hypothesis of the lemma, 
$${|\mathcal{G}|  } \geqslant q^2 {|\E\times \E|} \left(1- \frac{1}{q^w}\right), $$
and so, putting all together, we can conclude the desired result,
\begin{align*}  
    |\mathcal{G}_{\E}|&\geqslant \left((1-q^2)   +\left(q^2- \frac{1}{q^{  w-2}}\right)   \right) |\E\times \E|
    = \left(1- \frac{1}{q^{w-2}}\right) |\E\times \E|.
\end{align*}
\end{proof}

\begin{figure}[H]
\centering
\includegraphics{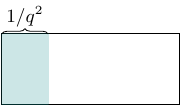}
\quad
\includegraphics{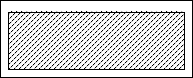}
\includegraphics{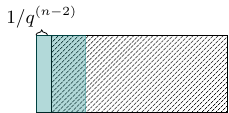}
\caption{The first image illustrates the ratio of $\E\times\E$ in $\Mon k \times \Mon k$, the second one is for the ratio of $ \mathcal{G}$ in $\Mon k \times \Mon k$ and the third one represents the \textit{worst case} intersection of the first two. }
\end{figure}

\begin{proof}[Proof of Theorem \ref{main thm*}]
We made the observation that the inequality
\begin{align}
    \dim  \< f_1, f_2, f_3\> \E >  2k+2 \label{main thm ineq1 } 
\end{align}  
holds true if $(f_2,f_3)$ satisfies the \emph{gcd condition} as given in \eqref{gcd_cond}.
Using Lemma \ref{main lemma}, and then taking $w=7$ in Lemma \ref{main lemma 2}, this probability is $\frac{|\mathcal{G}_{\E}|}{|\E \times \E|} \geqslant 1-\frac{1}{q^5}$. Following Remark \ref{reduction}, it means that $\frac{|\Psi|}{|\Psi \sqcup \Gamma|} \ge 1 - \frac{1}{q^5}$. Thus, if a triple $(f_1,f_2,f_3) \in \Psi \cup \Gamma$, then $(f_1,f_2,f_3) \in \Gamma$ with probability $\leqslant \frac{1}{q^5}$. For a triple $ (f_1,f_2,f_3) \in \R^3$, the probability of $(f_1,f_2,f_3) \in \Psi \cup \Gamma$ is $ 1-\frac{1}{q^3}$. Hence we get:
$$ \frac{|\E^3|}{|\E^3\cup \Gamma|} \geq \frac{\frac{1}{q^3}}{ \left( 1-\frac{1}{q^3} \right) \frac{1}{q^5} + \frac{1}{q^3}} \ge 1-\frac{1}{q^2}.$$  This, together with the subsequent statement complete the proof of the theorem.  
\end{proof}
\begin{corollary}\label{algorithm.polynomial prop}
    Let $f_1,f_2,f_3\in \R$ are three polynomials chosen uniformly at random. If $f_1,f_2\in \E$ and $ (f_1,f_2,f_3) \in \E^3\cup \Gamma $, then probability of $(f_1,f_2,f_3) \in \E^3$ is $\geqslant 1-\frac{1}{q^2}$.
  \end{corollary}

\begin{proof}
    Defining 
    $$\Gamma_0:= \{ (f_1,f_2,f_3) \in \Gamma : \dim\< f_1, f_2, f_3\> \R\leqslant 2k+2  \text{ and } f_1,f_2\in \E \}, $$
    and using $ \Gamma_0\subseteq \Gamma$ the conditional probability described in the proposition can be lower bounded by
    $$ \frac{|\E|}{ |\E \cup \Gamma_0|} \geqslant  \frac{|\E|}{ |\E \cup \Gamma|} \geqslant 1-\frac{1}{q^2}.$$
\end{proof}

\begin{remark}\label{recoveringE}
    After iterating the process of appending a polynomial to primarily found $(f_1, f_2, f_3) \in \E^3 \sqcup \Gamma$ $k-3$ times, we get a basis of the codimension 1 subspace $\E$ with probability $(1 - \frac{(k-3)}{q^2})$. 
\end{remark}
\section{Conclusion}
We showed that qGRS codes are easily distinguishable from random codes using
classical cryptanalysis techniques on algebraic codes. Moreover, 
we proposed a complete key-recovery attack on TGRS and actually qGRS
codes with $\ell = 1$. We left some questions open. First,
our proof of the validity of the attack rests on a probability analysis which 
does not permit us to treat some cases where the attack still probably works.
Second, for some very specific cases such as $h = 1$ or $k-2$, despite
being able to distinguish the codes, we cannot conclude the attack. We hope such cases
to be attacked but this would require an ad hoc manner to finish the attack.
Finally the major remaining question is the case of $\ell > 1$ twists. It is clear that a too high number of twist is not practical for cryptography since the decryption
complexity is exponential in the number of twists. Still, the question of small
$\ell \geq 2$ remains open, since the codes are distinguishable from random ones,
one can be optimistic on the possibility to extend the attack to larger $\ell$'s but this probably requires many technical adaptations both in the attack itself and in the probability analysis to prove the validity of the attack.

\bibliographystyle{splncs04}
\bibliography{Reference}

\begin{thebibliography}{10}
\providecommand{\url}[1]{\texttt{#1}}
\providecommand{\urlprefix}{URL }
\providecommand{\doi}[1]{https://doi.org/#1}

\bibitem{ABCCGLMMMNPPPSSSTW20}
Albrecht, M., Bernstein, D.J., Chou, T., Cid, C., Gilcher, J., Lange, T.,
  Maram, V., von Maurich, I., Mizoczki, R., Niederhagen, R., Persichetti, E.,
  Paterson, K., Peters, C., Schwabe, P., Sendrier, N., Szefer, J., Tjhai, C.J.,
  Tomlinson, M., Wen, W.: Classic {M}c{E}liece (merger of {Classic McEliece}
  and {NTS-KEM}). \url{https://classic.mceliece.org} (Nov 2022), fourth round
  finalist of the NIST post-quantum cryptography call

\bibitem{BBC+}
Baldi, M., Bianchi, M., Chiaraluce, F., Rosenthal, J., Schipani, D.: Enhanced
  public key security for the {M}c{E}liece cryptosystem. J. Cryptology
  \textbf{29},  1--27 (2016)

\bibitem{BBPR}
Beelen, P., Bossert, M., Puchinger, S., Rosenkilde, J.: Structural properties
  of twisted {R}eed-{S}olomon codes with applications to cryptography. In:
  Proc. IEEE Int. Symposium Inf. Theory - ISIT~2018. pp. 946--950. IEEE (2018)

\bibitem{BPR17}
Beelen, P., Puchinger, S., n{\'e}~Nielsen, J.R.: Twisted {R}eed-{S}olomon
  codes. In: Proc. IEEE Int. Symposium Inf. Theory - ISIT~2017. pp. 336--340.
  IEEE (2017)

\bibitem{Beelen}
Beelen, P., Puchinger, S., Rosenkilde, J.: Twisted {R}eed--{S}olomon codes.
  IEEE Trans. Inform. Theory  \textbf{68}(5),  3047--3061 (2022)

\bibitem{BenBen07}
Benjamin, A.T., Bennett, C.D.: The probability of relatively prime polynomials.
  Mathematics Magazine  \textbf{80}(3),  196--202 (2007).
  \doi{10.1080/0025570X.2007.11953481}

\bibitem{BL05}
Berger, T., Loidreau, P.: How to mask the structure of codes for a
  cryptographic use. Des. Codes Cryptogr.  \textbf{35},  63--79 (2005)

\bibitem{CCMZ15}
Cascudo, I., Cramer, R., Mirandola, D., Z{\'e}mor, G.: Squares of random linear
  codes. IEEE Trans. Inform. Theory  \textbf{61}(3),  1159--1173 (3 2015).
  \doi{10.1109/TIT.2015.2393251}

\bibitem{CGGOT}
Couvreur, A., Gaborit, P., Gauthier-Uma{\~n}a, V., Otmani, A., Tillich, J.P.:
  Distinguisher-based attacks on public-key cryptosystems using
  {R}eed--{S}olomon codes. Des. Codes Cryptogr.  \textbf{73},  641--666 (2014)

\bibitem{CL22}
Couvreur, A., Lequesne, M.: On the security of subspace subcodes of
  {R}eed--{S}olomon codes for public key encryption. IEEE Trans. Inform. Theory
   \textbf{68}(1),  632--648 (2022). \doi{10.1109/TIT.2021.3120440}

\bibitem{CLT19}
Couvreur, A., Lequesne, M., Tillich, J.P.: Recovering short secret keys of
  {RLCE} in polynomial time. In: Ding, J., Steinwandt, R. (eds.) Post-Quantum
  Cryptography~2019. LNCS, vol. 11505, pp. 133--152. Springer, Chongquing,
  China (May 2019)

\bibitem{CMP17}
Couvreur, A., M{\'a}rquez-Corbella, I., Pellikaan, R.: Cryptanalysis of
  {M}c{E}liece cryptosystem based on algebraic geometry codes and their
  subcodes. IEEE Trans. Inform. Theory  \textbf{63}(8),  5404--5418 (8 2017)

\bibitem{COT14a}
Couvreur, A., Otmani, A., Tillich, J.P.: Polynomial time attack on wild
  {M}c{E}liece over quadratic extensions. In: Nguyen, P.Q., Oswald, E. (eds.)
  Advances in Cryptology - EUROCRYPT~2014. LNCS, vol.~8441, pp. 17--39.
  Springer Berlin Heidelberg (2014)

\bibitem{COT17}
Couvreur, A., Otmani, A., Tillich, J.P.: Polynomial time attack on wild
  {M}c{E}liece over quadratic extensions. IEEE Trans. Inform. Theory
  \textbf{63}(1),  404--427 (1 2017)

\bibitem{COTG15}
Couvreur, A., Otmani, A., Tillich, J.P., Gauthier-Umana, V.: A polynomial-time
  attack on the {BBCRS} scheme. In: IACR International Workshop on Public Key
  Cryptography. pp. 175--193. Springer (2015)

\bibitem{FGVOPT13}
Faugère, J.C., Gauthier-Umaña, V., Otmani, A., Perret, L., Tillich, J.P.: A
  distinguisher for high-rate {M}c{E}liece cryptosystems. IEEE Trans. Inform.
  Theory  \textbf{59}(10),  6830--6844 (2013). \doi{10.1109/TIT.2013.2272036}

\bibitem{GS98}
Guruswami, V., Sudan, M.: Improved decoding of reed-solomon and
  algebraic-geometric codes. In: Proceedings 39th Annual Symposium on
  Foundations of Computer Science (Cat. No.98CB36280). pp. 28--37 (1998).
  \doi{10.1109/SFCS.1998.743426}

\bibitem{JM96}
Janwa, H., Moreno, O.: {McEliece} public key cryptosystems using
  algebraic-geometric codes. Des. Codes Cryptogr.  \textbf{8}(3),  293--307
  (1996)

\bibitem{KRW21}
Khathuria, K., Rosenthal, J., Weger, V.: Encryption scheme based on expanded
  {R}eed--{S}olomon codes. Adv. Math. Commun.  \textbf{15}(2),  207--218
  (2021). \doi{10.3934/amc.2020053},
  \url{http://aimsciences.org//article/id/0f055199-6fe4-404f-b206-517ce7d02a58}

\bibitem{LR}
Lavauzelle, J., Renner, J.: Cryptanalysis of a system based on twisted
  {R}eed--{S}olomon codes. Des. Codes Cryptogr.  \textbf{88}(7),  1285--1300
  (2020)

\bibitem{McEliece}
McEliece, R.: A public-key system based on algebraic coding theory. Jet
  Propulsion Laboratory, California Institute of Technology  \textbf{3},
  85--86 (1978)

\bibitem{MTSB13}
Misoczki, R., Tillich, J.P., Sendrier, N., Barreto, P.S.L.M.: {MDPC-McEliece}:
  New {McEliece} variants from moderate density parity-check codes. In: Proc.
  IEEE Int. Symposium Inf. Theory - ISIT. pp. 2069--2073 (2013).
  \doi{10.1109/ISIT.2013.6620590}

\bibitem{Nie86}
Niederreiter, H.: Knapsack-type cryptosystems and algebraic coding theory.
  Prob. Contr. Inform. Theory  \textbf{15}(2),  157--166 (1986)

\bibitem{Roth}
Roth, R.: Introduction to Coding Theory. Cambridge University Press, USA (2006)

\bibitem{Sheekey}
Sheekey, J.: A new family of linear maximum rank distance codes. Adv. Math.
  Commun.  \textbf{10}(3),  475--488 (2016). \doi{10.3934/amc.2016019},
  \url{https://www.aimsciences.org/article/id/3ff311ae-18e0-47d8-8edc-ab169dbbd975}

\bibitem{S94}
Sidelnikov, V.M.: A public-key cryptosytem based on {R}eed-{M}uller codes.
  Discrete Math. Appl.  \textbf{4}(3),  191--207 (1994)

\bibitem{SS92}
Sidelnikov, V.M., Shestakov, S.: On the insecurity of cryptosystems based on
  generalized {Reed-Solomon} codes. Discrete Math. Appl.  \textbf{1}(4),
  439--444 (1992)

\bibitem{sage}
Stein, W., et~al.: {S}age {M}athematics {S}oftware ({V}ersion 9.5). The Sage
  Development Team (2022), {\tt http://www.sagemath.org}

\bibitem{Sudan97}
Sudan, M.: Decoding of {R}eed--{S}olomon codes beyond the error-correction
  bound. Journal of Complexity  \textbf{13}(1),  180--193 (1997).
  \doi{https://doi.org/10.1006/jcom.1997.0439},
  \url{https://www.sciencedirect.com/science/article/pii/S0885064X97904398}

\bibitem{Wan16}
Wang, Y.: Quantum resistant random linear code based public key encryption
  scheme {RLCE}. In: Proc. IEEE Int. Symposium Inf. Theory - ISIT~2016. pp.
  2519--2523. {IEEE}, Barcelona, Spain (Jul 2016).
  \doi{10.1109/ISIT.2016.7541753},
  \url{http://dx.doi.org/10.1109/ISIT.2016.7541753}

\bibitem{Wan17}
Wang, Y.: {RLCE}--{KEM}. \url{http://quantumca.org} (2017), {F}irst round
  submission to the {NIST} post-quantum cryptography call

\bibitem{Wie}
Wieschebrink, C.: An attack on a modified {Niederreiter} encryption scheme. In:
  Yung, M., Dodis, Y., Kiayias, A., Malk, T. (eds.) Public-Key Cryptography -
  PKC~ 2006. LNCS, vol.~3958, pp. 14--26. Springer (2006)

\bibitem{Wie06}
Wieschebrink, C.: Two {NP}-complete problems in coding theory with an
  application in code based cryptography. In: Proc. IEEE Int. Symposium Inf.
  Theory - ISIT~2006. pp. 1733--1737. IEEE (2006)

\bibitem{Wie10}
Wieschebrink, C.: Cryptanalysis of the {Niederreiter} public key scheme based
  on {GRS} subcodes. In: Post-Quantum Cryptography~2010. LNCS, vol.~6061, pp.
  61--72. Springer (2010)

\end{thebibliography}

\end{document}